\documentclass[11pt,letterpaper]{article}
\usepackage[lmargin=1.0in,rmargin=1.0in,bottom=1.0in,top=1.0in,twoside=False]{geometry}

\usepackage{fullpage,amssymb,amsmath}
\usepackage{graphicx}
\usepackage{enumerate}
\usepackage[utf8]{inputenc}
\usepackage[T1]{fontenc}
\usepackage[dvipsnames]{xcolor}
\usepackage{amsfonts}
\usepackage{comment}
\usepackage[english]{babel}
\usepackage{libertine}
\usepackage[absolute]{textpos}
\usepackage{enumitem}
\usepackage{mathtools}
\usepackage{marvosym}
\usepackage{multicol}

\definecolor{MyBlue}{RGB}{50,100,200}
\usepackage[pdftex]{hyperref}
\hypersetup{%
  colorlinks=true,
  linkcolor=MyBlue,
  citecolor=MyBlue,
  urlcolor=MyBlue
}

\usepackage[amsmath,thmmarks,hyperref]{ntheorem}
\usepackage[title]{appendix}
\usepackage{cleveref}
\crefname{appsec}{Appendix}{Appendices}

\crefformat{page}{#2page~#1#3}%
\Crefformat{page}{#2Page~#1#3}%
\crefformat{equation}{#2(#1)#3}%
\Crefformat{equation}{#2(#1)#3}%
\crefformat{figure}{#2Figure~#1#3}%
\Crefformat{figure}{#2Figure~#1#3}%
\crefformat{section}{#2Section~#1#3}
\Crefformat{section}{#2Section~#1#3}
\crefformat{appsec}{#2Appendix~#1#3}
\Crefformat{appsec}{#2Appendix~#1#3}
\crefformat{chapter}{#2Chapter~#1#3}
\Crefformat{chapter}{#2Chapter~#1#3}
\crefformat{chapter*}{#2Chapter~#1#3}
\Crefformat{chapter*}{#2Chapter~#1#3}
\crefformat{part}{#2Part~#1#3}
\Crefformat{part}{#2Part~#1#3}
\crefformat{enumi}{#2(#1)#3}
\Crefformat{enumi}{#2(#1)#3}

\theoremnumbering{arabic}
\theoremstyle{plain}
\theoremsymbol{}
\theorembodyfont{\itshape}
\theoremheaderfont{\normalfont\bfseries}
\theoremseparator{.}

\newtheorem{theorem}{Theorem}
\crefformat{theorem}{#2Theorem~#1#3}
\Crefformat{theorem}{#2Theorem~#1#3}

\newcommand{\newtheoremwithcrefformat}[2]{%
  \newtheorem{#1}[theorem]{#2}%
  \crefformat{#1}{##2\MakeUppercase#1~##1##3}%
  \Crefformat{#1}{##2\MakeUppercase#1~##1##3}%
}
\newcommand{\newseptheoremwithcrefformat}[2]{%
  \newtheorem{#1}{#2}%
  \crefformat{#1}{##2\MakeUppercase#1~##1##3}%
  \Crefformat{#1}{##2\MakeUppercase#1~##1##3}%
}

\newtheoremwithcrefformat{lemma}{Lemma}
\newtheoremwithcrefformat{proposition}{Proposition}
\newtheoremwithcrefformat{observation}{Observation}
\newtheoremwithcrefformat{corollary}{Corollary}
\newseptheoremwithcrefformat{claim}{Claim}
\newseptheoremwithcrefformat{definition}{Definition}
\theorembodyfont{\upshape}
\newtheoremwithcrefformat{example}{Example}
\newseptheoremwithcrefformat{remark}{Remark}

\theoremstyle{nonumberplain}
\newseptheoremwithcrefformat{conjecture}{Conjecture}
\theoremheaderfont{\scshape}
\theorembodyfont{\normalfont}
\theoremsymbol{\ensuremath{\square}}
\newtheorem{proof}{Proof}

\theoremsymbol{\ensuremath{\lrcorner}}

\def\cqedsymbol{\ifmmode$\lrcorner$\else{\unskip\nobreak\hfil
\penalty50\hskip1em\null\nobreak\hfil$\lrcorner$
\parfillskip=0pt\finalhyphendemerits=0\endgraf}\fi}

\usepackage{tcolorbox}
\newenvironment{problem}[1]{\begin{tcolorbox}[title=\textsc{\color{black}#1},colback=white,colframe=Gray!40,enlarge left by=0mm,
        boxsep=5pt,
        arc=0pt,outer arc=0pt,]
}{\end{tcolorbox}}

\def\N{\mathbb{N}}
\def\O{\mathcal{O}}
\newcommand{\Oh}{\O}

\def\G{\mathcal{G}}

\def\S{\mathcal{S}}
\def\Ff{\mathcal{F}}

\DeclareMathOperator{\ctw}{\mathbf{ctw}}

\def\barcs{\overleftarrow{A}}

\DeclareMathOperator{\First}{\mathsf{start}}
\DeclareMathOperator{\Last}{\mathsf{end}}
\DeclareMathOperator{\first}{\mathsf{first}}
\DeclareMathOperator{\last}{\mathsf{last}}
\DeclareMathOperator{\head}{\mathsf{head}}
\DeclareMathOperator{\tail}{\mathsf{tail}}

\DeclareMathOperator{\markerX}{\mathsf{X}}
\DeclareMathOperator{\markerH}{\mathsf{H}}

\def\himmt{{\sc $H$-hitting Immersions in Tournaments}}

\DeclareMathOperator{\cut}{\mathsf{cut}}

\newcommand{\wh}[1]{\widehat{#1}}
\newcommand{\whH}{\wh{H}}

\renewcommand{\leq}{\leqslant}

\renewcommand{\geq}{\geqslant}

\def\dirE{\overrightarrow{A}}
\def\barcs{\overleftarrow{A}}
\def\wh#1{\widehat{#1}}
\def\sing{\partial}
\def\gen{\Gamma}
\def\flover{\wh{\partial}}
\def\scat#1{\widetilde{#1}}
\def\glue{\oplus}

\def\Cctw{c_{H}}

\usepackage{todonotes}

\begin{document}

\title{Polynomial kernel for immersion hitting in tournaments\thanks{This work is 
a part of projects that have received funding from the European Research Council (ERC) 
under the European Union's Horizon 2020 research and innovation programme, grant agreements No.~714704 (\L{}.~Bo\.zyk) and No.~677651 (Mi.~Pilipczuk).
}}

\author{\L{}ukasz Bo\.zyk\thanks{
  Institute of Informatics, University of Warsaw, Poland, \texttt{l.bozyk@uw.edu.pl}.
} 
\and
Micha\l{}~Pilipczuk\thanks{
  Institute of Informatics, University of Warsaw, Poland, \texttt{michal.pilipczuk@mimuw.edu.pl}.}
}

%\affil{Faculty of Mathematics, Informatics, and Mechanics,\\ University of Warsaw, Banacha 2, 02-097 Warszawa}

\begin{textblock}{20}(0, 12.6)
\includegraphics[width=40px]{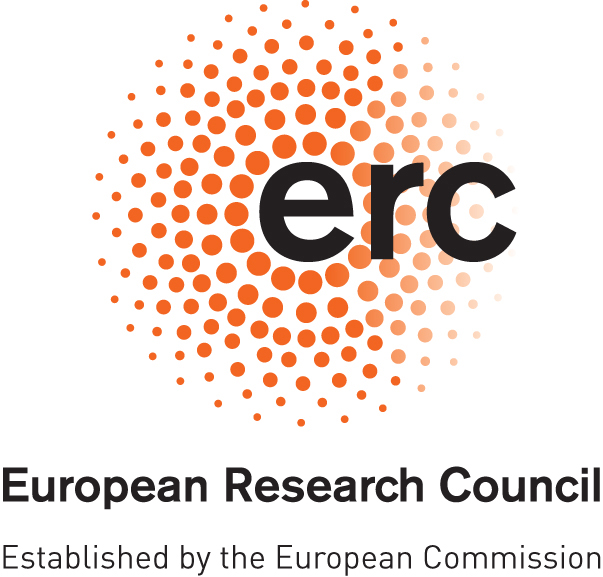}%
\end{textblock}
\begin{textblock}{20}(-0.25, 13)
\includegraphics[width=60px]{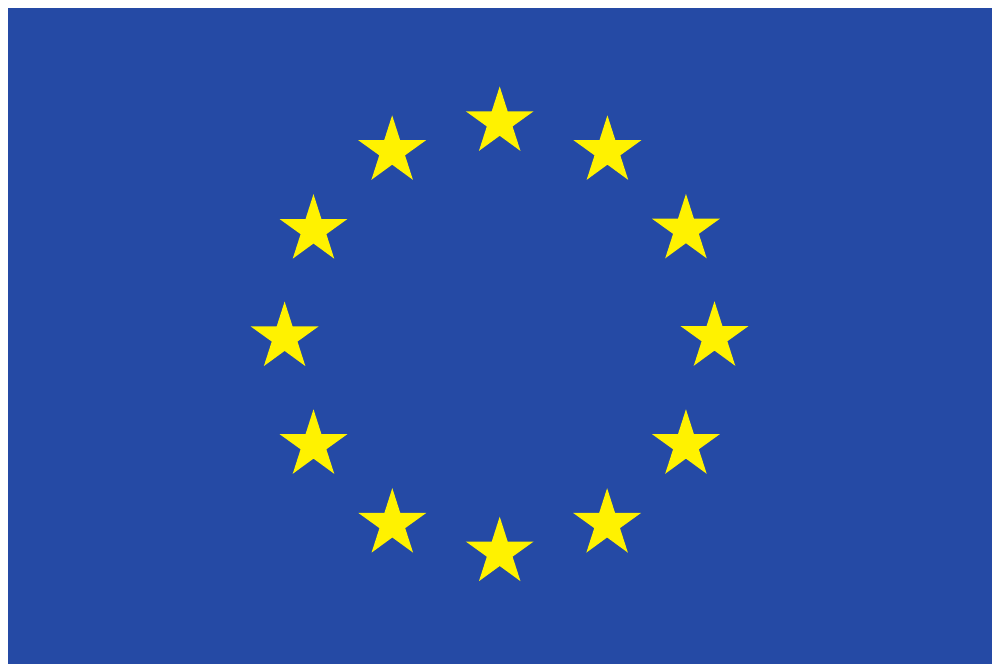}%
\end{textblock}

\maketitle

\begin{abstract}
\noindent For a fixed simple digraph $H$ without isolated vertices, we consider the problem of deleting arcs from a given tournament to get a digraph which does not contain $H$ as an immersion. We prove that for every~$H$, this problem admits a polynomial kernel when parameterized by the number of deleted arcs. The degree of the bound on the kernel size depends on $H$. 
\end{abstract}

\section{Introduction}

{\em{Kernelization}} is an algorithmic framework for describing preprocessing procedures that given an instance of a hard problem, identify and reduce easily resolvable parts. The usual formalization of the concept is based on the paradigm of parameterized complexity. A {\em{kernelization procedure}}, or a {\em{kernel}} for short, for a parameterized decision problem $L$ is a polynomial-time algorithm that given an instance $(I,k)$ of $L$, where $k$ is the parameter, outputs an equivalent instance $(I',k')$ such that both $|I'|$ and $k'$ are bounded by a computable function of $k$. If this function is a polynomial in $k$, we say that the kernel is {\em{polynomial}}. The search for (polynomial) kernels is an established research area within the field of parameterized algorithms. We refer to the textbook of Fomin et al.~\cite{squirrel} for a broader discussion of classic results and techniques.

A particularly fruitful line of research within kernelization concerns the methodology of {\em{protrusions}} and {\em{protrusion replacement}}. The idea is to find a large {\em{protrusion}}: a piece of graph that is simple --- for instance, has bounded treewidth --- and communicates with the rest of the graph only through a small interface. If found, a protrusion can be fully understood --- for instance, using dynamic programming on its tree decomposition --- and replaced with a smaller one with the same functionality. So if one proves that, provided the given instance is large, a large protrusion can be efficiently found and replaced with a strictly smaller one, then applying this strategy exhaustively eventually arrives at a kernel. Protrusion-based techniques were pioneered by Bodlaender et al.~\cite{BodlaenderFLPST16}, but by now have become a part of the standard toolbox of kernelization. We refer the interested reader to~\cite[Part~2]{squirrel} for more information.

A particularly important achievement in the development of protrusion-based kernelization procedures is the result of Fomin et al.~\cite{FominLMS12}, who gave a polynomial kernel for the {\sc{Planar $\Ff$-Deletion}} problem, defined as follows. Let $\Ff$ be a fixed family of graphs containing at least one planar graph. Then in the problem we are given a graph $G$ and an integer $k$ (considered to be the parameter), and the question is whether one can hit all minor models of graphs from $\Ff$ in $G$ using a hitting set consisting of at most $k$ vertices. Fomin et al. gave a polynomial kernel for this problem for every fixed family $\Ff$ as above. The degree of the polynomial bound on the kernel size depends on $\Ff$ and it is known that under certain complexity-theoretical assumptions, this is unavoidable~\cite{GiannopoulouJLS17}. The assumption that $\Ff$ contains a planar graph is crucial in the approach: under this assumption, graphs not containing any graph from $\Ff$ as a minor have treewidth bounded by a constant, which unlocks a multitude of tools related to tree decomposition using which one can understand the structure of the instance.

The concept of a protrusion, as described above, is quite capacious and can be applied in different settings as well. For instance, Giannopolou et al.~\cite{GiannopoulouPRT21} considered the {\sc{$\Ff$-Immersion Deletion}} problem, where for a given graph $G$ and parameter $k$, one wishes to hit all immersion models of graphs from $\Ff$ using a hitting set of edges of size at most $k$. (Recall that an {\em{immersion model}} of a graph $H$ in a graph $G$ consists of mapping the vertices of $H$ to distinct vertices of $G$ and edges of $H$ to pairwise edge-disjoint paths in $G$ so that the image of an edge $uv$ connects the image of $u$ with the image of $v$.) By loosely following the approach of~\cite{FominLMS12}, Giannopolou et al.~\cite{GiannopoulouPRT21} gave a {\em{linear}} kernel for {\sc{$\Ff$-Immersion Deletion}} for every family $\Ff$ that contains a subcubic planar graph. Here, the main idea was to adjust the notions of protrusions to the graph parameter {\em{tree-cutwidth}} and corresponding {\em{tree-cut decompositions}}, which play the same role for immersions as treewidth and tree decompositions play for minors.

Motivated by this, it is interesting to consider other settings where the protrusion methodology could be applied. A particularly tempting area is that of {\em{directed graphs}}, where natural analogues of problems considered in undirected graphs can be easily stated. Unfortunately, the structural theory of directed graphs is much less understood than that of undirected graphs, and many problems become inherently harder; see e.g.~\cite{GiannopoulouKKK20,GiannopoulouKKK22,KawarabayashiK15}. In particular, there is even a scarcity of fixed-parameter tractability results, not to mention kernelization results.

However, there is a particular class of directed graphs where a sound structural theory has been developed: {\em{tournaments}}. (Here, recall that a {\em{tournament}} is a directed graph where every pair of vertices is connected by exactly one arc). This theory\footnote{In this line of work, most results concern the class of {\em{semi-complete digraphs}}, which differ from tournaments by allowing that a pair of vertices can be also connected by two oppositely-oriented arcs. In this article we focus on the setting of tournaments for simplicity.} was pioneered by Chudnovsky, Ovetsky Fradkin, Kim, and Seymour~\cite{ChudnovskyFS12,ChudnovskyS11,FradkinS13,FradkinS15,KimS15}, while structural and algorithmic aspects connected to parameterized complexity were investigated by Fomin and Pilipczuk~\cite{FoPi}. See the introductory section of~\cite{FoPi} for an overview.

In particular, as proved in the aforementioned works, there are two main width notions for tournaments: {\em{cutwidth}} and {\em{pathwidth}}. The first one is tightly connected to (directed) immersions as follows: if a tournament $T$ excludes a digraph $H$ as an immersion, then the cutwidth of $T$ is bounded by a constant depending only on $H$. Pathwidth is connected to {\em{topological minors}} and {\em{strong minors}} in the same sense. These structural results were used for the design of parameterized algorithms for containment problems in tournaments in~\cite{ChudnovskyFS12,FradkinS13,FoPi}. Later, they were used by Raymond~\cite{Raymond18} and by Bożyk and Pilipczuk~\cite{BoPi} to establish Erd\H{o}s-P\'osa properties for immersions and topological minors in tournaments. 

The goal of this work is to explore the applicability of the structural theory of tournaments for kernelization, with a particular focus of developing a sound protrusion-based methodology.

\subparagraph*{Our contribution.}
For a simple directed graph $H$ without isolated vertices we define the following parameterized problem:

\begin{problem}{\himmt{}}
{\bf Input:} A tournament $T$ and a positive integer $k$.\\
{\bf Parameter:} $k$\\
{\bf Output:} Is there a set $F\subseteq A(T)$, such that $|F|\leq k$ and $T-F$ is $H$-immersion-free?
\end{problem}

That this problem is fixed-parameter tractable is proved in~\cite{FoPi}. 
Our main result states that for every fixed $H$, \himmt{} admits a polynomial kernel, of degree dependent on~$H$. Formally, we prove the following theorem.

\begin{theorem}\label{thm:main}
For every simple digraph $H$ without isolated vertices there exists a constant $c$ and an algorithm that given an instance $(T,k)$ of \himmt{}, runs in polynomial time and returns an equivalent instance $(T',k)$ with $|T'|\leq c\cdot k^{c}$.
\end{theorem}

We remark that when $H$ is the directed triangle, \himmt{} is equivalent to the {\sc{Feedback Arc Set in Tournaments}} ({\sc{FAST}}) problem. There is a sizeable literature on the parameterized complexity of {\sc{FAST}}, see e.g.~\cite{AlonLS09,Feige09,FoPi,KarpinskiS10}, mostly due to the fact that it admits a subexponential parameterized algorithm with running time $2^{\Oh(\sqrt{k})}\cdot n^{\Oh(1)}$. Kernelization procedures for {\sc{FAST}} were investigated by Bessy et al. in~\cite{BessyFGPPST11}, while kernelization of the dual problems of packing arc-disjoint triangles and packing arc-disjoint cycles in tournaments were recently studied by Bessy et al. in~\cite{BessyBKSSTZ21}.

On a very high conceptual level, the proof of \cref{thm:main} follows the classic blueprint of protrusion-based kernelization, like in e.g.~\cite{FominLMS12,GiannopoulouPRT21}. That is, if $(T,k)$ is a given instance of \himmt{}, we perform the following steps.
\begin{itemize}
 \item We may assume that the cutwidth of $T$ is bounded polynomially in $k$, for otherwise in $T$ one can find $k+1$ arc-disjoint immersion models of $H$; these witness a negative answer to the instance.
 \item Assuming that $T$ is large --- of size superpolynomial in $k$ --- but has cutwidth bounded polynomially in $k$, we may find in $T$ a large {\em{protrusion}}. Here, a protrusion is an interval $I$ in the vertex ordering $\sigma$ witnessing small cutwidth such that $\sigma$ restricted to $I$ witnesses that $T[I]$ has constant cutwidth, and there is only a constant number of $\sigma$-backward arcs with one endpoint in $I$ and the other outside of $I$. These are instantiations of the two desired properties of a protrusion: it has to have bounded width and communicate with the rest of the graph through a boundary of bounded size.
 \item We can replace the obtained protrusion with a strictly smaller one of the same ``type'', thus obtaining a strictly smaller equivalent instance. Applying this strategy exhaustively eventually yields a kernel of polynomial size.
\end{itemize}

Compared to the previous works, the main difficulty is to tame the interaction between a protrusion and the remainder of the instance. Namely, this interaction is not restricted to a set of vertices or arcs of constant size: as we work with tournaments, every vertex of a protrusion will necessarily have an arc connecting it to every vertex outside of the protrusion. The idea is that all but a constant number of those arcs will be forward arcs in the fixed vertex ordering $\sigma$ with bounded cutwidth. We call those well-behaved forward arcs {\em{generic}}, while the remaining constantly many backward arcs are {\em{singular}}. Understanding the interaction between a protrusion and the rest of the tournament as being governed by few singular arcs and a large number of well-behaved generic arcs is the crux of our approach.

In particular, while looking for a large replaceable protrusion, we have to be extremely careful when arguing about how such a protrusion may interact with optimum solutions. Here, a key step is to find several protrusions that appear consecutively in $\sigma$ (recall that our protrusions are intervals in $\sigma$), have the same {\em{type}} (in the sense of admitting partial immersions of $H$), and such that their union is a protrusion of again the same type. This step is done using Simon Factorization, a tool commonly used in the theory of automata and formal languages. Simon Factorization was recently used a few times in structural graph theory~\cite{BonamyP20,NesetrilMPRS21,NesetrilRMS20}, but we are not aware of any previous application in the context of kernelization.

The application of Simon Factorization is also the only step in the reasoning that causes the degree of the polynomial bounding the size of our kernel to depend on $H$. It is an interesting open question whether this can be improved, or in other words, whether there is a kernel of size at most $c\cdot k^d$, where $c$ may depend on $H$ but $d$ does not. Judging by the results on hitting immersions in undirected graphs~\cite{GiannopoulouPRT21}, we expect that this might be the case.

\section{Preliminaries}\label{sec:prelims}

We use the standard terminology and notation for describing immersions in tournaments and for cutwidth of digraphs and of tournaments. This terminology and notation is borrowed mostly, and often in a verbatim form, from the work of Bożyk and Pilipczuk~\cite{BoPi}.

For a positive integer $n$, we denote $[n]\coloneqq \{1,\ldots,n\}$ and $[-n]=\{-1,\ldots,-n\}$.

We use standard graph terminology and notation.
All graphs considered in this paper are finite, simple (i.e. without self-loops or multiple arcs with same head and tail), and directed (i.e. are {\em{digraphs}}). For a digraph $D$, by $V(D)$ and $A(D)$ we denote the vertex set and the arc set of $D$, respectively. We denote
\[|D|\coloneqq |V(D)|\qquad \textrm{and}\qquad \|D\|\coloneqq |A(D)|.\]
For $X\subseteq V(D)$, the subgraph {\em{induced}} by $X$, denoted $D[X]$, comprises of the vertices of $X$ and all the arcs of $D$ with both endpoints in $X$.
By $D-X$ we denote the digraph $D[V(D)\setminus X]$. Further, if $F$ is a subset of arcs of $D$, then by $D-F$ we denote the digraph obtained from $D$ by removing all the arcs of $F$.
For $X,Y\subseteq V(D)$ we denote by $\dirE(X,Y)$ the set of all arcs $(u,v)\in A(D)$ such that $u\in X$ and $v\in Y$ and moreover $A(X,Y)\coloneqq \dirE(X,Y)\cup\dirE(Y,X)$. For an arc $a=(u,v)\in A(D)$ we  define $\tail(a)\coloneqq u$ and $\head(a)\coloneqq v$. For a directed (not necessarily simple) path $P$ we denote by $\first(P)$ and $\last(P)$ the first and the last arcs on path $P$, respectively.

\subparagraph*{Tournaments.}
A simple digraph $T=(V,A)$ is called a \emph{tournament} if for every pair of distinct vertices $u,v\in V$, either $(u,v)\in A$, or $(v,u)\in A$ (but not both). Alternatively, one can represent the tournament $T$ by providing a pair $(\sigma,\barcs_\sigma(T))$, where $\sigma\colon V\to [|V|]$ is an {\em{ordering}} of the set $V$ and $\barcs_{\sigma}(T)$ is the set of \emph{$\sigma$-backward arcs}, that is, \[\barcs_{\sigma}(T)\coloneqq\{\,(u,v)\in A\ \mid\ \sigma(u)>\sigma(v)\,\}.\] All the remaining arcs are called \emph{$\sigma$-forward}. If the choice of ordering $\sigma$ is clear from the context, we will call the arcs simply \emph{backward} or \emph{forward}. For $\alpha,\beta\in \{0,1,\ldots,|V|\}$, $\alpha\leq \beta$, we define \[\sigma(\alpha,\beta]\coloneqq \{v\in V\mid \alpha<\sigma(v)\leq\beta\}.\] Sets $\sigma(\alpha,\beta]$ as defined above will be called \emph{$\sigma$-intervals}. 

\subparagraph*{Cutwidth.}
Let $T=(V,A)$ be a tournament and $\sigma$ be an ordering of $V$. %If $I=\sigma[\alpha,\beta]$, we denote $$\first_\sigma(I)\coloneqq \alpha\qquad\textrm{and}\qquad \last_\sigma(I)\coloneqq \beta.\] Moreover, let $\sigma[\alpha]\coloneqq \sigma(0,\alpha]$ and call this interval an \emph{$\alpha$-prefix} of $\sigma$. 
The set \[\cut_\sigma[\alpha]=\{(u,v)\in A\mid \sigma(u)>\alpha\geq\sigma(v)\}\subseteq\barcs_\sigma(T)\] is called the \emph{$\alpha$-cut} of $\sigma$. The {\em{width}} of the ordering $\sigma$ is equal to $\max_{0\leq\alpha\leq |V|}|\cut[\alpha]|$, and the \emph{cutwidth} of~$T$, denoted $\ctw(T)$, is the minimum width among all orderings of $V$.

It is perhaps a bit surprising that in tournaments, there is a very simple algorithm to compute an ordering of optimum width: just sort the vertices by outdegrees.

\begin{lemma}[see~\cite{BarberoPP18,fradkin-thesis}]\label{lem:sorting}
 Let $T$ be a tournament and $\sigma$ be an ordering of $T$ satisfying the following for every pair of vertices $u$ and $v$: if $u$ appears before $v$ in $\sigma$, then the outdegre of $u$ is not smaller than the outdegree of~$v$. Then the width of $\sigma$ is equal to $\ctw(T)$.
\end{lemma}

If $I=\sigma(\alpha,\beta]$, then we denote
\[\sing^+(I)\coloneqq \dirE(I,\sigma(0,\alpha])\quad\text{and}\quad \sing^-(I)\coloneqq \dirE(\sigma(\beta,|V|],I).\]
Note that $\sing^+(I)\subseteq \cut_\sigma[\alpha]$ and $\sing^-(I)\subseteq \cut_\sigma[\beta]$ and therefore $|\sing^+(I)|\leq c$ and $|\sing^-(I)|\leq c$, where $c$ is the width of $\sigma$. These inclusions may be proper, as the arcs from the set $\flover(I)\coloneqq \dirE(\sigma(\beta,|V|],\sigma(0,\alpha])$ contribute to the cuts but are not incident with $I$.
We define $\sing(I)\coloneqq\sing^+(I)\cup\sing^-(I)$ and call the elements of $\sing(I)$ \emph{$I$-singular} (or simply \emph{singular}) arcs.
Moreover, we define
\[\gen^+(I)\coloneqq\dirE(I,\sigma(\beta,|V|]),\qquad \gen^-(I) \coloneqq \dirE(\sigma(0,\alpha],I),\quad \textrm{and}\quad \gen(I)=\gen^+(I)\cup\gen^-(I),\]
and call the elements of $\gen(I)$ \emph{$I$-generic} (or simply \emph{generic}) arcs.

If $I'=V-I$ where $I=\sigma(\alpha,\beta]$, then we call the set $I'$ a \emph{co-interval} and define \emph{$I'$-singular} and \emph{$I'$-generic} arcs as follows
\[\sing^-(I')\coloneqq \sing^+(I),
%\subseteq \sing^-(\sigma(0,\alpha]),
\quad \sing^+(I')\coloneqq\sing^-(I),
%\subseteq\sing^+(\sigma(\beta,|V|]),
\quad \sing(I')\coloneqq\sing^+(I')\cup\sing^-(I')=\sing(I), \]
\[\gen^-(I')\coloneqq\gen^+(I),
%\subseteq\gen^-(\sigma(\beta,|V|]),
\quad\gen^+(I')\coloneqq\gen^-(I),
%\subseteq\gen^+(\sigma(0,\alpha]),
\quad \gen(I')\coloneqq \gen^-(I')\cup\gen^+(I')=\gen(I).\]

\subparagraph*{Immersions.}
A digraph $\wh{H}$ is an \emph{immersion model} (or briefly a \emph{copy}) of a digraph $H$ if there exists a mapping $\phi$, called an {\em{immersion embedding}}, such that: 
\begin{itemize}
\item vertices of $H$ are mapped to pairwise different vertices of $\whH$;
\item each arc $(u,v)\in A(H)$ is mapped to a directed path in $\whH$ starting at $\phi(u)$ and ending at $\phi(v)$; and
\item each arc of $\whH$ belongs to exactly one of the paths $\{\phi(a)\colon a\in A(H)\}$.
\end{itemize}
If it does not lead to misunderstanding, we will sometimes slightly abuse the above notation by identifying $\phi$ and $\wh{H}$ and calling $\phi$ the immersion model of $H$.
%If the immersion embedding $\phi$ is clear for the context, then for a subgraph $C$ of $H$ we define $\whH|_C$ to be the subgraph of $\whH$ consisting of all the vertices and arcs participating in the image of $C$ under $\phi$. Note that thus, $\whH|_C$ is an immersion model of $C$.

Let $H$ be a digraph.
We say that a digraph $G$ \emph{contains $H$ as an immersion} (or $H$ can be \emph{immersed} in $G$) if $G$ has a subgraph that is an immersion model of $H$. Digraph $G$ is called \emph{$H$-immersion-free} (or \emph{$H$-free} for brevity) if it does not contain $H$ as an immersion.

We will use the following result of Fomin and Pilipczuk~\cite{FoPi}.
%and the second author.

\begin{theorem}[Theorem~7.3 of~\cite{FoPi}]\label{thm:ctw-bound}
Let $T$ be a tournament which does not contain a digraph $H$ as an immersion. Then $\ctw(T)\in \O((|H|+\|H\|)^2)$.
\end{theorem}

\begin{corollary}\label{cor:const-ctw}
For every digraph $H$ there exists a constant $\Cctw$ such that for every $H$-free tournament $T$, we have $\ctw(T)\leq \Cctw$. 
\end{corollary}

\pagebreak[3]

Throughout this paper we fix a simple digraph $H$ without isolated vertices and an integer $k\in \N$.
For a tournament $T$, a set $F\subseteq A(T)$ is called a \emph{solution} if $T-F$ is $H$-free. Moreover, $F$ is an \emph{optimal solution} if it is a solution of the smallest possible size. So $(T,k)$ is a YES-instance of \himmt{} if and only if in $T$ there exists an optimal solution of size at most $k$.

\subparagraph*{Monoids and Simon factorization.}
Simon factorization was originally developed by Simon in~\cite{SIMON} and the currently best bounds are due to Kufleitner~\cite{Kufleitner08}. See also the work of Boja\'nczyk~\cite{BojanczykFF} for a nice exposition; we mostly follow the notation from that source.

Let $S$ be a finite monoid (i.e., a finite set equipped with an associative binary operation $\cdot$ and a neutral element $1$).  An element $e\in S$ is called \emph{idempotent} if $e\cdot e=e$. For a finite alphabet $A$, by $A^\star$ we denote the set of all finite words over $A$, and a {\em{morphism}} $\alpha\colon A^\star\to S$ is a function satisfying $\alpha(\varepsilon)=1$ ($\varepsilon$ being the empty word) and $\alpha(w_1w_2)=\alpha(w_1)\cdot\alpha(w_2)$ for every $w_1,w_2\in A^\star$. Note that a morphism is uniquely defined by the images of single symbols from $A$.
The following lemma is a direct consequence of Simon Factorization, see \cref{sec:appa} for a derivation.

\begin{lemma}\label{lem:Simon}
Let $S$ be a finite monoid, $A$ be a finite alphabet, and $\alpha\colon A^\star\to S$ be a morphism.
Suppose $w\in A^\star$ is a word of length at least $\ell^{3|S|}$. Then there exists a subword $w'$ of $w$ and an idempotent $e\in S$ such that $w'=w_1w_2\ldots w_\ell$, where $w_i\in A^\star$ are nonempty subwords of $w$ and 
\[\alpha(w_1)=\alpha(w_2)=\ldots=\alpha(w_\ell)=e.\]
\end{lemma}

Note that in the setting of \cref{lem:Simon}, given a word $w\in A^\star$ of length $n\geq \ell^{3|S|}$, one can easily find $w'$ and a suitable decomposition $w'=w_1w_2\ldots w_\ell$ in time $\Oh(|S|\cdot n^3)$ assuming unit cost of operations in $S$. Indeed, one can guess $e$ (by trying at most $|S|$ possibilities) and the first position of $w'$ within $w$ (by trying $n$ possibilities), and then for every subword $w''$ starting at this position compute the longest possible decomposition of the form $w''=w_1w_2\ldots w_{\ell'}$ such that $\alpha(w_1)=\ldots=\alpha(w_{\ell'})=e$, if existent. The latter can be done by a standard left-to-right dynamic programming in time $\Oh(n^2)$.

\section{Partial immersions}\label{sec:partial-immersions}

Our goal in this section is to extend the notion of an immersion to {\em{partial immersions}}. These will be used to understand possible behaviors of immersion models in a tournament $T$ with respect to different intervals in an ordering of the vertex set of $T$.
Let then $T=(V,A)$ be a tournament and let $\sigma$ be an ordering of $V$. For now, fix a $\sigma$-interval $I\coloneqq\sigma(\alpha,\beta]$.

%\subparagraph*{Partial immersions.}

\begin{definition}
A \emph{scattered path} in $I$ of {\em{size}} $q\geq 0$ is a sequence $\scat{P}=(P_i)_{i=1}^q$ satisfying the following properties:
\begin{itemize}
\item for each $i\in [q]$, $P_i$ is a directed (simple) path of length at least $1$ consisting of arcs that belong to $A(T[I])\cup\sing(I)\cup\gen(I)$;
\item paths $P_i$ for $i\in [q]$ are pairwise arc-disjoint;
\item for every $i\in [q]$, $i\neq 1$, we have $\first(P_i)\in \gen^-(I)\cup\sing^-(I)$;
\item for every $i\in [q]$, $i\neq q$, we have $\last(P_i)\in \gen^+(I)\cup\sing^+(I)$.
\end{itemize}
Each term in the sequence $(P_i)_{i=1}^q$ will be called a \emph{piece} of $\scat{P}$. The set of arcs of all pieces of $\scat{P}$ is denoted $A(\scat{P})$. If $\first(P_i)\in\gen^-(I)$ and $\last(P_i)\in\gen^+$, then the piece $P_i$ is called \emph{generic}.
\end{definition}
Note that $\first(P_1)$ is allowed to be an arc in the set $A(T[I])\cup\gen^+(I)\cup \sing^+(I)$. If it is such, we call the vertex $\tail(\first(P_1))\in I$ the \emph{beginning} of $\scat{P}$ and denote it by $\First(\scat{P})$. %; otherwise we set $\First(\scat{P})=\markerX$.
Similarly, if $\last(P_q)\in A(T[I])\cup \gen^-(I)\cup\sing^-(I)$, then we call the vertex $\head(\last(P_q))$ the \emph{end} of $\scat{P}$ and denote it by $\Last(\scat{P})$. %; otherwise we put $\Last(\scat{P})=\markerX$. Moreover, put $|\scat{P}|=j$. 
Note that the empty sequence is a scattered path of size $0$. Also, a scattered path with only one piece, whose both beginning and end exist, is just a path in $T[I]$.
By $\mathcal{P}_I$ we denote the family of all scattered paths in~$I$.

We say that a scattered path $\scat{P}=(P_i)_{i=1}^q$ in $I$ can be \emph{shortened} to a scattered path $\scat{P}'$ (or that $\scat{P}'$ is a \emph{shortening} of $\scat{P}$) if:
\begin{itemize}
\item $\First(\scat{P})=\First(\scat{P}')$ and $\Last(\scat{P})=\Last(\scat{P}')$ (meaning either equal or simultaneously undefined);
\item for each piece $P'_s$ of $\scat{P}'$ there exist indices $i_s^-\leq i_s^+$ such that $\tail(\first(P'_s))=\tail(\first(P_{i_s^-}))$ and $\head(\last(P'_s))=\head(\last(P_{i_s^+}))$; and
\item whenever $s<s'$, we have $i_s^+<i_{s'}^-$ and $P'_s$ appears before $P'_{s'}$ in $\scat{P}'$.
\end{itemize}
Intuitively, shortening of the path means removing some pieces and replacing several contiguous subsequences of the pieces with single pieces, keeping the tail of the beginning and the head of the end of the replaced subsequence. Note that in particular, some pieces of $\scat{P}$ can be simply omitted in $\scat{P}'$ (other than the initial and the terminal one).

\begin{definition}
A \emph{partial immersion embedding} of $H$ in $I$ (or shortly, a \emph{partial immersion} in $I$) is a mapping $\phi\colon A(H)\to \mathcal{P}_I$ such that
\begin{itemize}
\item all scattered paths $\phi(a)$ for $a\in A(H)$ are pairwise arc-disjoint;
\item if $\tail(a)=\tail(a')$ then $\First(\phi(a))$ and $\First(\phi(a'))$ are either equal, or simultaneously undefined;
\item if $\First(\phi(a))$ and $\First(\phi(a'))$ are defined and equal, then $\tail(a)=\tail(a')$;
\item if $\head(a)=\head(a')$ then $\Last(\phi(a))$ and $\Last(\phi(a'))$ are either equal, or simultaneously undefined;
\item if $\Last(\phi(a))$ and $\Last(\phi(a'))$ are defined and equal, then $\head(a)=\head(a')$.
\end{itemize}
\end{definition}
Intuitively, we can think of a partial immersion as of a ``trace'' which some immersion model $\wh{H}$ of $H$ in $T$ ``leaves'' on the interval $I$. Some edges of $H$ have images being paths in $\wh{H}$ non-incident with $I$ (these correspond to empty scattered paths in the partial immersion embedding). Some images of arcs of $H$ come back and forth to $I$, intersecting with $I$ along a non-empty scattered path (the ordering of paths on a single scattered path corresponds to the order of their appearance along the image of the respective arc of $H$). Finally, some arc images begin or end within $I$, which corresponds to the case when the beginning or the end of a scattered path is defined and is a vertex of $I$.

We call a partial immersion $\phi'$ in $I$ a \emph{shortening} of $\phi$ in $I$ if for every $a\in A(H)$, the scattered path $\phi'(a)$ is a shortening of $\phi(a)$. We call $\phi$  \emph{minimal} if there is no  shortening of $\phi$ with at least one scattered path of strictly smaller size. Note that $\phi$ may be minimal even if some piece of some $\phi(a)$ can be replaced by a different single piece with equal first and last vertices. Shortening which does not decrease the size of any scattered path will be called \emph{trivial}.

Note that each immersion model $\phi$ of $H$ in $T$ is a partial immersion in $V(T)$, in which all scattered paths $\phi(a)$, $a\in A(H)$, are paths in $T$ beginning and ending at $\phi(\tail(a))$ and $\phi(\head(a))$, respectively. Moreover, each partial immersion $\phi$ in $I$ gives rise to a natural partial immersion of $H$ in $J\subseteq I$ in which all paths $\phi(a)$ where $a\in A(H)$ are ``trimmed'' to scattered paths consisting of precisely those arcs which are incident with $J$. %We will call such natural partial immersion a \emph{trace of $\phi$ on $J$}.

Formally, let $I$ and $J$ be $\sigma$-intervals such that $J\subseteq I$. If $P$ is a path in $I$, then define the \emph{trace} $P|_{J}$ of $P$ on $J$ to be the scattered path consisting of all arcs of $P$ incident with $J$, arranged in the order of appearance along $P$. If $\scat{P}=(P_i)_{i=1}^q$ is a scattered path in $I$, then define the \emph{trace $\scat{P}|_{J}$} of $\scat{P}$ on $J$ to be the concatenation of scattered paths $P_i|_J$. The \emph{trace} $\phi|_{J}$ of a partial immersion $\phi$ of $H$ in $I$ on $J$ is defined by setting $\phi|_{J}(a)=(\phi(a))|_{J}$ for every $a\in A(H)$.

Consider any $\sigma$-intervals $I_1$, $I_2$ with the property that there exist $\alpha,\beta,\gamma\in \{0,1,\ldots,|V|\}$ such that $I_1=\sigma(\alpha,\gamma]$ and $I_2=\sigma(\gamma,\beta]$. We will call such a pair of intervals \emph{consecutive}. Equivalently, two disjoint intervals are consecutive  if their union $I=I_1\cup I_2$ is a $\sigma$-interval as well. Let $\scat{P}_1$ be a scattered path in $I_1$ and $\scat{P}_2$ be a scattered path in $I_2$. We say that $\scat{P}_1$ and $\scat{P}_2$ are \emph{compatible} if
\begin{itemize}
\item $A(\scat{P}_1)\cap A(I_1,I_2)=A(\scat{P}_2)\cap A(I_1,I_2)$;
\item the set of pieces whose arc set is $A(\scat{P}_1)\cup A(\scat{P}_2)$ can be ordered to form a scattered path $\scat{P}$ in $I$ with the property that all pieces of $\scat{P_i}$ appear in $\scat{P}$ in the same order as they do in $\scat{P_i}$, for $i=1,2$.
\end{itemize}
Every $\scat{P}$ described above will be called a \emph{gluing} of $\scat{P}_1$ and $\scat{P}_2$. Note that a gluing is not necessarily uniquely defined, which can be seen particularly well when $A(\scat{P}_1)\cap A(I_1,I_2)=A(\scat{P}_2)\cap A(I_1,I_2)=\emptyset$ --- in this case the pieces of both paths can be  ``shuffled'' in any way only keeping the order of pieces originating from the same path.

\begin{observation}
If $\scat{P}_1$ is compatible with $\scat{P}_2$ and $\scat{P}'_2$ is a shortening of $\scat{P}_2$, then there exists a shortening $\scat{P}'_1$ of $\scat{P}_1$ such that $\scat{P}'_1$ and $\scat{P}'_2$ are compatible and every gluing of them is a shortening of some gluing of $\scat{P}_1$ and~$\scat{P}_2$.  
\end{observation}
\begin{proof}
To construct $\scat{P}'_1$ from $\scat{P}_1$ it is enough to omit the pieces which do not share arcs with $\scat{P}'_2$. 
\end{proof}

We will say that two partial immersions $\phi_1$ in $I_1$ and $\phi_2$ in $I_2$ are \emph{compatible} if there exists a partial immersion $\phi$ in $I$ such that $\phi_1=\phi|_{I_1}$ and $\phi_2=\phi|_{I_2}$, or --- in other words --- that for every $a\in A(H)$ the scattered path $\phi(a)$ is a gluing of $\phi_1(a)$ and $\phi_2(a)$. We will call every such $\phi$ a \emph{gluing} of $\phi_1$ and $\phi_2$. Note that a gluing is not necessarily uniquely defined. Denote the set of all gluings of $\phi_1$ and $\phi_2$ by $\phi_1\glue\phi_2$.

\begin{observation}\label{obs:gluing-minimal}
If $\phi\in\phi_1\glue\phi_2$ and $\phi$ is minimal, then $\phi_1$ and $\phi_2$ are minimal.
\end{observation}

\begin{proof}
Suppose that $\phi_1$ is not minimal. Then there exists a shortening $\phi'_1$ of $\phi_1$ which has a strictly smaller total number of pieces. For every $a\in A(H)$ the scattered path $\phi'_1(a)$ is compatible with some scattered path $\phi'_2(a)$ obtained by omitting several pieces in $\phi_2(a)$. In particular, the collection of all such $\phi'_2(a)$ --- call it $\phi'_2$ --- is a partial immersion compatible with $\phi'_1$ that is a shortening of $\phi_2$. Any gluing belonging to $\phi'_1\glue\phi'_2$ is a shortening of $\phi$ with a strictly smaller total number of pieces --- contradiction. An analogous reasoning can be applied upon supposing that $\phi_2$ is not minimal. 
\end{proof}

%\ukryj{\begin{proof}
%Suppose that $\phi_1$ is not minimal. Then there exists a shortening $\phi'_1$ of $\phi_1$ which has a strictly smaller total number of pieces. For every $a\in A(H)$ the scattered path $\phi'_1(a)$ is compatible with some scattered path $\phi'_2(a)$ obtained by omitting several pieces in $\phi_2(a)$. In particular, the collection of all such $\phi'_2(a)$ --- call it $\phi'_2$ --- is a partial immersion compatible with $\phi'_1$ that is a shortening of $\phi_2$. Any gluing belonging to $\phi'_1\glue\phi'_2$ is a shortening of $\phi$ with a strictly smaller total number of pieces --- contradiction. An analogous reasoning can be applied upon supposing that $\phi_2$ is not minimal. 
%\end{proof}}

The notions of a scattered path, partial immersion and trace can be naturally extended to co-intervals, by applying all definitions verbatim. If $I$ is a $\sigma$-interval and $I'=V-I$ is the corresponding co-interval, then partial immersions $\phi_1$ in $I$ and $\phi_2$ in $I'$ are \emph{compatible} if there exists an immersion $\phi$ in $T$ such that $\phi_1=\phi|_{I}$ and $\phi_2=\phi|_{I'}$. Again, every such $\phi$ is called a \emph{gluing} of $\phi_1$ and $\phi_2$.

\subparagraph*{Types of intervals.} The key ingredient of our analysis is a constant-size encoding of the set of possible ``behaviors'' of partial immersions in intervals. 

A $\sigma$-interval $I$ shall be called \emph{$\ell$-long} if $|I|\geq \ell$. Further, we shall call $I$ \emph{$c$-flat} if $|\sing^+(I)|\leq c$, $|\sing^-(I)|\leq c$, and $\sigma$ restricted to $T[I]$ has width at most $c$.

Note that if $I=\sigma(\alpha,\beta]$ is $2r$-long, then the intervals $I^-_r\coloneqq \sigma(\alpha,\alpha+r]$ and $I^+_r\coloneqq \sigma(\beta-r,\beta]$ are disjoint. On the other hand, if $I$ is~$c$-flat, then we can color all backward arcs incident with $I$ with at most $3c$ colors in such a way that each $\gamma$-cut of $\sigma$ restricted to those arcs contains arcs of mutually different colors. This can be achieved e.g. by greedy coloring the $\gamma$-cuts for consecutive $\gamma=\alpha,\ldots,\beta+1$. Formally, there exists a function $\xi\colon \barcs_\sigma(T)\cap A(I,V)\to [3c]$ such that for every $\gamma\in[|V|-1]$ and every two distinct arcs $a_1,a_2\in \cut_\sigma[\gamma]\cap A(I,V)$ we have $\xi(a_1)\neq\xi(a_2)$. In the following fix such a function.

\begin{definition}
Let $\phi$ be a partial immersion in an interval $I$ that is $2r$-long and $c$-flat. For each $a\in A(H)$ we define the $(r,c)$-\emph{type} $\tau^{(r,c)}(\phi(a))$ of the scattered path $\phi(a)=(P_i)_{i=1}^q$ as the following sequence of length~$2q$: 
\[(f_-(\first(P_1)),f_+(\last(P_1)),f_-(\first(P_2)),f_+(\last(P_2)),\ldots,f_-(\first(P_q)),f_+(\last(P_q))),\] where the functions $f_{\pm}\colon A(T[I])\cup \sing(I)\cup\gen(I)\to [-3c]\cup[r]\cup\{\markerX,\markerH\}$ are defined as follows

\[f_{-}(a)=\begin{cases}
-\xi(a)&\text{ if }a\in\sing^{-}(I),\\
\sigma(\head(a))-\alpha&\text{ if }a\in\gen^{-}(I)\text{ and }\head(a)\in I^-_r,\\
\markerX&\text{ if }a\in\gen^{-}(I)\text{ and }\head(a)\notin I^-_r,\\
\markerH&\text{ otherwise;}
\end{cases}\]
\[f_{+}(a)=\begin{cases}
-\xi(a)&\text{ if }a\in\sing^{+}(I),\\
r+\sigma(\tail(a))-\beta&\text{ if }a\in\gen^{+}(I)\text{ and }\tail(a)\in I^+_r,\\
\markerX&\text{ if }a\in\gen^{+}(I)\text{ and }\tail(a)\notin I^+_r,\\
\markerH&\text{ otherwise.}
\end{cases}\]
Let $S$ be the set of all terms of the sequences $(\tau^{(r,c)}(\phi(a)))_{a\in A(H)}$, where by term we mean a value with an assigned position (i.e. equal values in different sequences are considered different terms).
The \emph{$(r,c)$-type} of $\phi$ is the collection of types $\tau^{(r,c)}(\phi)=(\tau^{(r,c)}(\phi(a)))_{a\in A(H)}$ equipped with a pair of equivalence relations $(R_-,R_+)$ on the set $S\cup[3c]$ defined as follows.
\begin{itemize}
\item If a piece $P_1$ of $\phi(a_1)$ and a piece $P_2$ of $\phi(a_2)$ satisfy $\head(\first(P_1))=\head(\first(P_2))$, then the corresponding terms in $\tau^{(r,c)}(\phi)$ (elements of the set $\{\markerX\}\cup [r]$ with assigned positions in respective sequences) are in relation $R_-$.
\item If a piece $P$ of $\phi(a)$ is such that $\head(\first(P))$ is the tail of a singular arc of color $x\in[3c]$, then the corresponding term in $\tau^{(r,c)}(\phi)$ is in relation $R_-$ with $x$.
\item Analogously, if $\tail(\last(P_1))=\tail(\last(P_2))$, then the terms $f_+(\last(P_1))$ and $f_+(\last(P_2))$ are in relation $R_+$. And $x\in [3c]$ is in relation with all terms corresponding to tails of ends of paths which are simultaneously the head of the singular arc of color $x$.
\end{itemize}
\end{definition}

Let us provide some intuition on what kind of information is stored in the type of $\phi$ defined above. First of all, the entire ``singular interface'' of this partial immersion is kept, i.e. in the type we remember precisely the singular arcs used to enter or exit $I$ when traversing along each scattered path $\phi(a)$. Moreover, if we enter or exit $I$ with a  generic arc, we remember the precise vertex of entry/exit inside $I$, but only if it is close enough to the ``border'' of $I$ (i.e. within the first or last $r$ vertices in $\sigma$). Otherwise we remember the respective entry/exit as ``generic'', which is marked by the marker $\markerX$. Moreover, we keep the information about whether the scattered path begins or ends inside $I$ --- the marker $\markerH$ represents that a vertex of $H$ is mapped under the partial immersion embedding to a vertex within $I$. Finally if the  extreme (first or last) arcs of some pieces are generic and have the same first/last vertex in $I$, we remember this fact in the equivalence relations $R_\pm$. We shall need this information to be able to ``glue'' two partial immersions without using the same generic gluing arc for different pieces. The incidence with singular arcs is also stored to avoid a situation when one attempts a generic gluing  along a singular arc. The reason for colors being stored as negative integers is purely technical --- it ensures that $[r]\cap [-3c]=\emptyset$.

An \emph{$(r,c)$-type} is any collection of sequences of even lengths over the set $[-3c]\cup[r]\cup\{\markerX,\markerH\}$ indexed by $A(H)$ and equipped with a pair of equivalence relations $(R_-,R_+)$ on the union of the set of all terms of these sequences and $[3c]$. An \emph{$I$-admissible $(r,c)$-type} is every type $\tau$ such that there exists a partial immersion $\phi$ in $I$ such that $\tau=\tau^{(r,c)}(\phi)$. 
The \emph{size} of a type is the sum of lengths of its sequences (i.e. it does not depend on the equivalence relations).

Let $\gamma_r$ be a function mapping each element from the set $[r]$ onto the single marker $\markerX$, and identity otherwise. Intuitively, the function $\gamma_r$ keeps only the information about generic nature of the end of a piece, and ``forgets'' about the closeness of this vertex to the boundary. We say that an $(r,c)$-type $\tau'$ is a \emph{shortening} of an $(r,c)$-type $\tau$ if for every $a\in A(H)$ the sequence $\gamma(\tau'_a)$ is a subsequence of $\gamma(\tau_a)$ and the two sequences have the same first and last terms. 

The following observation is a direct consequence of the definition of shortening of a partial immersion.

\begin{observation}
If a partial immersion $\phi'$ of $(r,c)$-type $\tau'$ is a shortening of a partial immersion $\phi$ of $(r,c)$-type~$\tau$, then $\tau'$ is a shortening of $\tau$.
\end{observation}

\begin{definition}
An $I$-admissible $(r,c)$-type $\tau$ is called \emph{minimal} (in $I$) if there is no $I$-admissible type $\tau'$ of strictly smaller size that would be a  shortening of $\tau$.
\end{definition}

\begin{observation}\label{obs:minimal-types}
Let $\phi$ be a partial immersion in $I$ of $(r,c)$-type $\tau$. Then $\phi$ is minimal if and only if $\tau$ is minimal in $I$.
\end{observation}
\begin{proof}
If $\phi$ is not minimal, then for any nontrivial shortening $\phi'$ of $\phi$, $\tau^{(r,c)}(\phi')$ is a shortening of $\tau$ of strictly smaller size.
Conversely, if $\tau$ is not minimal, then there exists an $I$-admissible type $\tau'$ of strictly smaller size. Every partial immersion $\phi'$ of type $\tau'$ is a nontrivial shortening of $\phi$.
\end{proof}

We now observe that in a minimal partial immersion on an interval $I$, all scattered paths will be relatively small, that is, will visit $I$ only a bounded number of times.

\begin{lemma}\label{lem:teleports}
Suppose that $F\subseteq A(T)$ has the property that $|F|\leq f$, $I$ is $4\|H\|(c+f+1)$-long and $c$-flat, and a partial immersion $\phi$ in $I$ is minimal and disjoint with $F$. Then for every $a\in A(H)$ the scattered path $\phi(a)$ has size at most $2c+3$.
%the length of $\tau^r(\phi)$ is at most $4c+6$.
\end{lemma}

\begin{proof}
We will show that every partial immersion $\phi$ disjoint with $F$ can be shortened to a partial immersion $\phi'$ disjoint with $F$ in which every path $\phi(a)$ has at most one generic piece. This will yield the desired result as every nongeneric piece of $\phi(a)$ which is neither the initial, nor the terminal one, contains at least one arc from the set $\sing(I)$, which is of size not exceeding $2c$ (as $I$ is $c$-flat).

Take an arbitrary $\phi$ disjoint with $F$. For every $a\in A(H)$, let $\phi(a)=(P^a_i)_{i=1}^{q_a}$. Let $i_a^-$ be the smallest index $i$ such that $\first(P^a_i)\in\gen^-(I)$ (if there is no such index, we put $i_a^-=\infty$). Similarly, let $i_a^+$ be the greatest index $i$ such that $\last(P^a_i)\in\gen^+(I)$ (or $i_a^+=-\infty$ if no such index exists). Note that if $i<i_a^-$ or $i>i_a^+$, then $P^a_i$ is not generic, so
\[\left|\bigcup_{a\in A(H)}\bigcup_{i\notin [i_a^-,i_a^+]} \{P^a_i\}\right|\leq \|H\|\cdot (2c+2).\]
The set on the left-hand side of the above inequality comprises of pieces that we will keep unchanged in~$\phi$; we will show a way to shorten $\phi$ by replacing the subsequence $(P^a_i)_{i=i_a^-}^{i_a^+}$ with a single generic piece for each $\phi(a)$ for which this subsequence is nonempty (if $i^-_a > i^+_a$, then $\phi(a)$ has no generic pieces). This will conclude the proof as the new piece will be the only generic one.

Let 
\[\G=\gen(I)\cap \bigcup_{a\in A(H)}\bigcup_{i\notin [i_a^-,i_a^+]} A(P^a_i),\]
so in particular $|\G|\leq \|H\|\cdot (2c+2)$ as each piece $P^a_i$ with $i\notin [i_a^-,i_a^+]$ contains at most one generic arc in $\gen(I)$. Since $I$ is $2\|H\|(c+f+1)$-long, we may find in $I$ a subset $J$ of vertices not incident with $F\cup \G\cup\sing(I)$ such that $|J|\geq \|H\|$. Indeed,
\[|I|\geq 4\|H\|(c+f+1)\geq  2f+2(c+1)\|H\|+2c+\|H\|\geq 2|F|+|\G|+|\sing(I)|+\|H\|.\]
Now assign to each $a\in A(H)$ a different vertex $v(a)\in J$. We can use these vertices are ``pivots'' for the newly constructed generic pieces. Indeed, for every $a\in A(H)$ such that $i_a^-\leq i_a^+$ consider the path
\[P^a_{\pm}=(\tail(\first(P^a_{i_a^-})),v(a))(v(a),\head(\last(P^a_{i_a^+}))).\]
This path is well-defined as the two arcs it consists of are generic (because $v(a)$ was chosen not to be incident with $\sing(I)$). Moreover, it is edge-disjoint with $\G\cup F$ and with all other $P^{a'}_{\pm}$ ($a'\neq a$), as $v(a)\neq v(a')$. Inserting paths $P^a_\pm$ in place of the subsequences described above gives a desired shortening and finishes the proof.
\end{proof}

We call an $(r,c)$-type $\tau$ \emph{short} if for every $a\in A(H)$ the length of $\tau(\phi(a))$ is at most $4c+6$.

\begin{corollary}\label{cor:const-types}
For $r\geq 1$ we have the following:
\begin{enumerate}
\item If an interval $I$ is $4\|H\|(r+c)$-long and $c$-flat, and a partial immersion $\phi$ in $I$ is minimal, then $\tau^{(r,c)}(\phi)$ is short. 
\item For every pair of integers $r,c\in \N$ there is a constant $t(r,c)$ such that there are exactly $t(r,c)$ short $(r,c)$-types. In particular, for any interval $I$ as above, there are at most $t(r,c)$ different $I$-admissible minimal $(r,c)$-types. 
\end{enumerate}
\end{corollary}

\begin{proof}
It is enough to take $F=\emptyset$ in \cref{lem:teleports} to see that the lengths of sequences $\tau^{(r,c)}(\phi(a))$ for $a\in A$ are uniformly bounded by $2(2c+3)$. To prove that the number of types is bounded in terms of $r$ and $c$ it is now enough to observe that the number of different equivalence classes $R_+$, $R_-$ is so. The last claim is true as these are the equivalence classes introduced on a set of size bounded in terms of $r$ and $c$. 
\end{proof}

%\ukryj{\begin{proof}
%It is enough to take $F=\emptyset$ in \cref{lem:teleports} to see that the lengths of sequences $\tau^{(r,c)}(\phi(a))$ for $a\in A$ are uniformly bounded by $2(2c+3)$. To prove that the number of types is bounded in terms of $r$ and $c$ it is now enough to observe that the number of different equivalence classes $R_+$, $R_-$ is so. The last claim is true as these are the equivalence classes introduced on a set of size bounded in terms of $r$ and $c$. 
%\end{proof}}

Let $I_1$, $I_2$ be two consecutive $c$-flat, $2(r+c)$-long $\sigma$-intervals and $I=I_1\cup I_2$. We say that two $(r,c)$-types $\tau_1$ of $I_1$ and $\tau_2$ of $I_2$ are \emph{compatible} if there exists a partial immersion $\phi$ in $I$ such that $\phi|_{I_1}$ and $\phi|_{I_2}$ are compatible, $\phi|_{I_1}$ has type $\tau_1$, and $\phi|_{I_2}$ has type $\tau_2$.

The \emph{gluing} of the types $\tau_1\glue\tau_2$ is defined as the set of all $(r,c)$-types of all such $\phi$.  The following lemma proves that this definition is correct, i.e. that types of two immersions store enough information to ensure the possibility of gluing them.

\begin{lemma}\label{lem:types-behave-well}
Let $I_1$, $I_2$ be two consecutive $c$-flat, $4\|H\|(r+c)$-long $\sigma$-intervals and $I=I_1\cup I_2$. If two $(r,c)$-types $\tau_1$ of $I_1$ and $\tau_2$ of $I_2$ are \emph{compatible}, then for every partial immersion $\phi_1$ in $I_1$ of type $\tau_1$ and for every partial immersion $\phi_2$ in $I_2$ of type $\tau_2$, the immersions $\phi_{1}$ and $\phi_{2}$ are compatible.
\end{lemma}

\begin{proof}
By the definition of compatible types, we have at least one pair of compatible partial immersions $\eta_1$ in $I_1$ and $\eta_2$ in $I_2$, of types $\tau_1$ and $\tau_2$ respectively. Consider any $\eta\in \eta_1\glue\eta_2$. We will use $\eta$ as a ``model'' for understanding how to successfully glue $\phi_1$ and $\phi_2$.

Fix $a\in A(H)$ and mappings $\pi_i$ sending the pieces of path $\eta_i(a)$ to the corresponding pieces of path $\phi_i(a)$, in accordance with the equality of types, where $i=1,2$. Moreover, color in $\eta(a)$ all arcs contained in $A(I_1,I_1)\cup\sing^+(I_1)\cup\gen^-(I_1)$ blue, all arcs contained in $A(I_2,I_2)\cup\sing^-(I_2)\cup\gen^+(I_2)$ green, and all other arcs (i.e. the ``gluing'' arcs in $A(I_1,I_2)$) red. Thus, every piece in $\eta(a)$ consists of alternating (and possibly empty) green and blue paths with a red arc between each two consecutive paths of different colors. 

Suppose that some piece $P$ of $\eta(a)$ contains a red arc $e$. Assume that the maximal monochromatic nonred subpath of $P$ directly preceding this arc is blue and the maximal monochromatic nonred subpath of $P$ directly following this arc is green (the reasoning in the other case is analogous). Then the blue path is a subpath of some piece $P_1$ of $\eta_1$ and the green path is a subpath of some piece $P_2$ of $\eta_2$, where the only difference is the arc $e$ and possibly the other extreme arcs (if they are colored red as well). We will show that we may directly glue pieces $\pi_1(P_1)$ of $\phi_1$ and $\pi_2(P_2)$ of $\phi_2$ using the arc 
\[e'=(\tail(\last(\pi_1(P_1))),\head(\first(\pi_2(P_2))))\] and moreover --- that such gluing can be performed simultaneously in all places corresponding to red arcs of $\eta$. This yields a construction of a gluing $\phi\in \phi_1\glue \phi_2$ therefore certifying the desired compatibility. Moreover, $\tau^{(r,c)}(\phi)=\tau^{(r,c)}(\eta)$ and $\eta$ was chosen arbitrarily, which implies that the types of immersions being glued are enough to determine all possible types of gluings.

First of all note that if the arc $e$ is singular of color $x$, then this fact is stored in the equivalence relations $R_+$ of $\tau^{(r,c)}(\eta_1)$ and $R_-$ of $\tau^{(r,c)}(\eta_2)$. Therefore, by the equality of types of $\eta_i$ and $\phi_i$, the arc $e'$ is singular of color $x$ as well, hence $e'=\last(\pi_1(P_1))=\first(\pi_2(P_2))$ and the gluing is possible. 

If $e$ is generic then, again by the respective equivalence relations, no other generic gluing arc joins vertices $\tail(\last(P_1))$ and $\head(\last(P_1))$. This means that $e'$ is generic and does not correspond to any other red arc of $\eta$ than $e$. Therefore again the gluing along this arc can be performed.
%Fix $a\in A(H)$. Note that the pieces $\phi_1(a)$ and $\phi_2(a)$ sharing singular edges can be glued, because the same singular edges were used for gluing $\eta_1$ and $\eta_2$ (which in turn follows from $\tau^{(r,c)}(\phi_i)=\tau^{(r,c)}(\eta_i)$, $i=1,2$ and the fact that the type stores the identity of extreme singular arcs of scattered paths). Moreover, the equality of equivalence relations (which are a part of the type) ensures that the direct generic gluings are possible as well (as they were possible in $\eta$).
\end{proof}

\subparagraph*{Boundaried intervals and signatures.}

In the process of replacing protrusions we will need to consider intervals not as subsets of an ordered vertex set of a larger tournament, but as standalone structures which can be used to replace one another. We introduce the notion of an $(r,c)$-boundaried interval to enable such considerations. 

\begin{definition}
An \emph{$(r,c)$-boundaried interval}%\mipicom{I don't get why $\Sigma$ is treated as a second coordinate in the definition. Shoulnd't we just have $\Sigma$ to be always the set of immersions present in $D$?}\lbcom{Keeping it as it is for now, not to mess something up accidentally} 
is a digraph $D$ on vertex set $V(D)=S^+\cup I\cup S^-$ equipped with an ordering $\sigma_I$ of $I$. Furthermore, we require the following:
\begin{itemize}
\item $|I|\geq 4\|H\|(r+c)$;
\item $D[I]$ is a tournament;
\item $\sigma_I$ has width at most $c$;
\item $S^-\subseteq [3c]\times\{-\}$ and $S^+\subseteq [3c]\times\{+\}$;
\item each vertex of $S^+$ has only one incident arc and this arc belongs to the set $\dirE(I,S^+)$;
\item each vertex of $S^-$ has only one incident arc and this arc belongs to the set $\dirE(S^-,I)$.
%\item $\dirE(D-\{v^-\},\{v^-\})=\dirE(\{v^+\},D-\{v^+\})=\emptyset$;
%\item the vertices $v^\pm$ are incident to all vertices in $I$ and to each of them with a multiple arc of size $2c+3$.
\end{itemize}
The \emph{$r$-boundary} of $D$ is the pair of sets $I^-$ and $I^+$ consisting of the first and the last $r$ vertices of $I$ in $\sigma_I$, respectively. For $D$ defined above, we will shortly write $D=I\cup S^\pm$.
\end{definition}

Note that this notion emulates a $4\|H\|(r+c)$-long $c$-flat interval $I$ in the following sense. Arcs whose heads are contained in $S^+$ correspond to $\sing^+(I)$, and arcs whose tails are contained in the set $S^-$ correspond to $\sing^-(I)$. The names of the auxiliary vertices in $S^\pm$ correspond to the $\xi$-colors of the respective backward arcs. The $r$-boundary consists of precisely those vertices whose generic entry or exit is remembered in the $(r,c)$-type of a partial immersion.

Formally, every $4\|H\|(r+c)$-long $c$-flat $\sigma$-interval $I$ in $T$ can be uniquely encoded with an $(r,c)$-boundaried interval $D^{(r,c)}(I)$, whose structure resembles the structure of $T[I]$ and $\sing(I)$ as follows:
\begin{itemize}
\item the vertices and arcs of $I$ are kept along with their ordering in $T$;
\item every singular arc $a\in\sing^\pm(I)$ is mapped to an arc joining $(\xi(a),\pm)\in S^\pm$ with the  endpoint of $a$ contained in $I$;
\item the projection of $S^\pm$ onto the first coordinate is precisely $\{\xi(a)\mid a\in \sing^\pm(I)\}$.
%\item $\Sigma^{(r,c)}(I)$ is the set of all $I$-admissible short $(r,c)$-types.
\end{itemize}

The notion of a partial immersion can be naturally adjusted to the setting of boundaried intervals. The only difference is the lack of the ``generic interface'' i.e. there are no auxiliary edges in boundaries intervals used to emulate $\gen(I)$. These can be, however, emulated by storing the information from the type (marker $\markerX$ or number in $[r]$ if the generic arc is incident with the $r$-boundary, and the equivalence relations $R_{\pm}$) instead of the identity of particular generic arcs. Formally, a piece of a scattered path in $(r,c)$-boundaried interval can begin or end with an element in $\{\markerX\}\cup[r]$ instead of a generic arc. In particular, this slightly modified variant of partial immersions can be equipped with precisely the same definition of admissible $(r,c)$-type as in the former case.

\begin{definition}\label{def:signature}
An \emph{$(r,c)$-signature} is a subset of the set of all short $(r,c)$-types.
The \emph{$(r,c)$-signature} $\Sigma^{(r,c)}(I)$ of a $4\|H\|(r+c)$-long $c$-flat $\sigma$-interval $I$ is the set of all $I$-admissible minimal $(r,c)$-types.
The \emph{$(r,c)$-signature} $\Sigma^{(r,c)}(D)$ of an $(r,c)$-boundaried interval $D=I\cup S^\pm$ is the set of all $I$-admissible minimal $(r,c)$-types.
\end{definition}

The intuition behind this definition is that if $I$ is appropriately long and flat, then the signature of $I$ stores the information about all possible interactions of $I$ with minimal partial immersions. Note that $\Sigma^{(r,c)}(D^{(r,c)}(I))=\Sigma^{(r,c)}(I)$.

We say that two $(r,c)$-boundaried intervals $I\cup S^\pm$ and $I'\cup S'^\pm$ are \emph{exchangeable} if they have equal $(r,c)$-signatures, $S^\pm=S'^\pm$ (both intervals use precisely the same colors on the boundary), and the incidence structure of $r$-boundaries of those intervals with backward arcs is the same, i.e. for every $i\in [r]$ the set of colors of singular arcs incident with both $S^\mp$ and the $i$-th vertex of $I^{\pm}$ is the same as analogously defined set of colors for $i$-th vertex of $I'^{\pm}$. Intuitively this means that in $T$ we may replace the interval $I$ with $I'$.

\begin{corollary}\label{cor:const-sign}
For every pair of integers $r,c\in \N$ there is a constant $s(r,c)$ such that there are exactly $s(r,c)$ different $(r,c)$-signatures.
\end{corollary}
\begin{proof}
We may set $s(r,c)= 2^{t(r,c)}$, where $t(r,c)$ is the constant provided by \cref{cor:const-types}.
\end{proof}

For future discussion of algorithmic aspects, we will need the following observation.

\begin{lemma}\label{lem:signature-compute}
 Consider $r$ and $c$ fixed and let $T$, $\sigma$, and $I$ be as in \cref{def:signature}. Then given $T$, $\sigma$, and $I$, the signature $\Sigma^{(r,c)}(I)$ can be computed in polynomial time.
\end{lemma}

\begin{proof}
 It suffices to decide, for every short $(r,c)$-type $\tau$, whether $\tau$ is $I$-admissible. This can be done in polynomial time using the algorithm for rooted immersion containment of Fomin and Pilipczuk~\cite{FoPi}. Alternatively, it follows from the result of Ganian et al.~\cite[Theorem~17]{GanianHO11} that $I$-admissibility of a fixed type $\tau$ can be expressed using an $\mathsf{MSO}$ formula. Then one can apply the algorithm for $\mathsf{MSO}$ model-checking on structures of bounded cliquewidth~\cite{CourcelleMR00} in conjunction with the observation that boundedness of cutwidth entails boundedness of cliquewidth for tournaments~\cite{FoPi}.
\end{proof}
%\ukryj{\begin{proof}
% It suffices to decide, for every short $(r,c)$-type $\tau$, whether $\tau$ is $I$-admissible. This can be done in polynomial time using the algorithm for rooted immersion containment of Fomin and Pilipczuk~\cite{FoPi}. Alternatively, it follows from the result of Ganian et al.~\cite[Theorem~17]{GanianHO11} that $I$-admissibility of a fixed type $\tau$ can be expressed using an $\mathsf{MSO}$ formula. Then one can apply the algorithm for $\mathsf{MSO}$ model-checking on structures of bounded cliquewidth~\cite{CourcelleMR00} in conjunction with the observation that boundedness of cutwidth entails boundedness of cliquewidth for tournaments~\cite{FoPi}.
%\end{proof}}

Let $\S^{(r,c)}$ be the set of all $(r,c)$-signatures; we have $|\S^{(r,c)}|=s(r,c)$, where $s(r,c)$ is the constant given by \cref{cor:const-sign}. Let $\S^{(r,c)}_\sigma$ be the set of those $(r,c)$-signatures which are equal to $\Sigma^{(r,c)}(I)$ for some $\sigma$-interval $I$, i.e. contain only $I$-admissible minimal $(r,c)$-types. Then $\S^{(r,c)}_\sigma\subseteq S^{(r,c)}$, so $|\S^{(r,c)}_\sigma|\leq s(r,c)$. We argue that $\S^{r,c}_\sigma$ has a structure of a semigroup in the following sense.

\begin{lemma}\label{lem:composability}
Let $I_1$, $I_2$ be two $4\|H\|(r+c)$-long $c$-flat $\sigma$-intervals such that $I=I_1\cup I_2$ is a $c$-flat $\sigma$-interval. Then $\Sigma^{(r,c)}(I)$ is uniquely determined by $\Sigma^{(r,c)}(I_1)$ and $\Sigma^{(r,c)}(I_2)$. 
\end{lemma}

\begin{proof}
It is enough to prove that:
\begin{itemize}[nosep]
\item[(i)] for every $\tau\in \Sigma^{(r,c)}(I)$ there exist $\tau_1\in\Sigma^{(r,c)}(I_1)$ and $\tau_2\in\Sigma^{(r,c)}(I_2)$ which are compatible and $\tau\in\tau_1\glue\tau_2$; and
\item[(ii)] for every two compatible $\tau_1\in\Sigma^{(r,c)}(I_1)$ and $\tau_2\in\Sigma^{(r,c)}(I_2)$ and every $\tau\in\tau_1\glue\tau_2$ that is minimal in $I$, we have $\tau\in \Sigma^{(r,c)}(I)$.
\end{itemize}

For the proof of (i) take a minimal partial immersion $\phi$ in $I$ with $\tau^{(r,c)}=\tau$. We know that $\tau\in\tau^{(r,c)}(\phi|_{I_1})\oplus\tau^{(r,c)}(\phi|_{I_2})$, so it is enough to prove that $\phi|_{I_1}$ and $\phi|_{I_2}$ are minimal. But this follows directly from \cref{obs:gluing-minimal}.

For the proof of (ii) note that $\tau_1\glue\tau_2$ consists only of $I$-admissible types.
\end{proof}

%\ukryj{\begin{proof}
%It is enough to prove that:
%\begin{itemize}
%\item[(i)] for every $\tau\in \Sigma^{(r,c)}(I)$ there exist $\tau_1\in\Sigma^{(r,c)}(I_1)$ and $\tau_2\in\Sigma^{(r,c)}(I_2)$ which are compatible and $\tau\in\tau_1\glue\tau_2$; and
%\item[(ii)] for every two compatible $\tau_1\in\Sigma^{(r,c)}(I_1)$ and $\tau_2\in\Sigma^{(r,c)}(I_2)$ and every $\tau\in\tau_1\glue\tau_2$ that is minimal in $I$, we have $\tau\in \Sigma^{(r,c)}(I)$.
%\end{itemize}
%
%For the proof of (i) take a minimal partial immersion $\phi$ in $I$ with $\tau^{(r,c)}=\tau$. We know that $\tau\in\tau^{(r,c)}(\phi|_{I_1})\oplus\tau^{(r,c)}(\phi|_{I_2})$, so it is enough to prove that $\phi|_{I_1}$ and $\phi|_{I_2}$ are minimal. But this follows directly from \cref{obs:gluing-minimal}.
%
%For the proof of (ii) note that $\tau_1\glue\tau_2$ consists only of $I$-admissible types.
%\end{proof}}

\Cref{lem:composability} implies that the set $\S^{(r,c)}_\sigma\cup \{0\}$ can be endowed with an associative binary product operation such that for every two consecutive intervals $I_1$, $I_2$, the product of their signatures is the signature of their union $I_1\cup I_2$. Formally, we set the product for all pairs of consecutive intervals as above; \cref{lem:composability} enures that this is well-defined. Next, for all pairs of elements $\tau_1,\tau_2\in \S^{r,c}_\sigma$ for which their product is not yet defined, we set $\tau_1\cdot \tau_2=0$. Also, we set $0=0\cdot 0=0\cdot \tau=\tau\cdot 0$ for all $\tau\in \S^{(r,c)}_\sigma$. In this way, $\S^{(r,c)}_\sigma\cup \{0\}$ becomes a monoid; the empty signature is the neutral element of multiplication.

By \cref{lem:Simon} we obtain the following.

\begin{corollary}\label{cor:Simon-intervals}
Suppose $I$ is a $c$-flat $4\|H\|(r+c)\ell^{3s(r,c)}$-long $\sigma$-interval. Then there exists a sequence of consecutive $4\|H\|(r+c)$-long $c$-flat $\sigma$-intervals $(I_i)_{i=1}^\ell$ whose $(r,c)$-signatures are equal and equal to the signature of their union. Moreover, given $r$, $c$, $T$, $\sigma$, and $I$, such a sequence can be found in polynomial time.
\end{corollary}

\begin{proof}
Let $(X_i)_{i=1}^{\ell^{3s(r,c)}}$ be an arbitrary sequence of consecutive $4\|H\|(r+c)$-long $c$-flat $\sigma$-intervals included in $I$. Note that $\Sigma^{(r,c)}_\sigma\colon A^\ast\to \S^{(r,c)}$ is a properly defined morphism, sending a single interval to its signature and a concatenation of two consecutive intervals to the unique signature of their union (which is not $0$).  Applying \cref{lem:Simon} finishes the proof; the algorithmic statement follows from the remark after \cref{lem:Simon}.
\end{proof}

\section{Finding protrusions}\label{sec:finding}

In order to find an appropriately large subgraph of $T$ which does not ``affect'' the behavior of $T$ with respect to \himmt{}, we roughly proceed as follows. First, we find a suitable ordering $\sigma$ of $V(T)$ and an appropriately long interval $X$ in $\sigma$ such that $X$ has a constant-size singular interface towards the remainder of $T$. Then, inside $X$, we find (again, an appropriately long) subinterval $I$ of a very specific structure: $I$ can be divided into $2k+3$ subintervals with the same signatures as itself. This is where we use Simon Factorization through \cref{cor:Simon-intervals}. In the next part we use this extra structure to prove that one of these subintervals can be replaced with a strictly smaller replacement in such a way that after the substitution, we obtain an equivalent instance of the problem.

We proceed to a formal implementation of this plan. The first lemma gives the ordering $\sigma$ and the interval~$X$.

\begin{lemma}\label{lem:protrusion-const-boundary}
Let $T$ be a tournament with $\ctw(T)\leq c$ and $|T|\geq (2c+1)(x+1)(k+1)$. If $T$ contains at most $k$ arc-disjoint immersion copies of $H$, then there exists an ordering $\sigma$ of $V(T)$ %of width at most $3c$ 
and an $H$-free $\sigma$-interval $X$ such that $|X|\geq x$ and $X$ is $\Cctw$-flat with respect to $\sigma$.

Moreover, given $T,k,c,x$ as above, one can in polynomial time either conclude that $T$ contains more than $k$ arc-disjoint immersion copies of $H$, or find an ordering $\sigma$ and an interval $X$ satisfying the above properties.
\end{lemma}

\begin{proof}
We first argue the existential statement, and then address the algorithmic aspects at the end.

Let $\sigma_0$ be an ordering of $V(T)$ of width $\ctw(T)\leq c$.
Consider any $k+1$ vertex-disjoint $\sigma_0$-intervals, each containing at least $(2c+1)(x+1)$ vertices. By assumption, at least one of them, say $I$%=\sigma_0(\alpha,\beta]$
, induces an $H$-free graph. Let %$B\subseteq I$ be the set of $\sigma$-backward arcs with exactly one endpoint in $I$, i.e. 
$B=\sing(I)$. Then $|B|\leq 2c$.

Note that we have $\ctw(T[I])\leq \Cctw$ by \cref{cor:const-ctw}. Let $\sigma_I$ be an ordering of vertices of $I$ of width $\ctw(T[I])$ and let $\sigma$ be the ordering obtained %\mipicom{We never argue that the width of $\sigma$ is optimal.} 
from $\sigma_0$ by reordering $I$ according to $\sigma_I$, that is,
\begin{itemize}[nosep]
\item If $u,v\in I$ then $\sigma(u)<\sigma(v)\iff \sigma_I(u)<\sigma_I(v)$.
\item If $u\notin I$ or $v\notin I$, then  $\sigma(u)<\sigma(v)\iff \sigma_0(u)<\sigma_0(v)$.
\end{itemize}
%Note that if $\gamma\leq \alpha$ or $\gamma\geq \beta$, then $|\cut_\sigma[\gamma]|=|\cut_{\sigma_0}[\gamma]|\leq c$. If $\alpha<\gamma<\beta$, then every $\gamma$-cut in $\sigma$ contains at most $\ctw(T)-|\wh{\sing}(I)|$ $\sigma$-backward arcs from $A(I,I)$. Hence \[|\cut_\sigma[\gamma]|\leq \ctw(T[I])+|\wh{\sing}(I)|+|B|\leq 3c,\] which proves that the width of $\sigma$ is not greater than $3c$.

Consider the set 
\[I_{-B}\coloneqq I-\left(\{\tail(a) \mid a\in \sing^+(I)\}\cup \{\head(a) \mid a\in \sing^-(I)\} \right),\]
i.e. the vertices in $I$ which are not incident with $B$. Since $|B|\leq 2c$, we have that $|I_{-B}|>(2c+1)x$ and $I_{-B}$ is the union of at most $2c+1$ $\sigma$-intervals.
Hence, at least one of these intervals has size at least $x$. Call it $X$.

Let \[I^-\coloneqq \{v\in I~|~\sigma(v)<\sigma(w)\text{ for every }w\in X\}\quad\textrm{and}\quad I^+=I-I^--X.\] Note that as no arc from $B$ is adjacent to $X$, $\dirE(X,I^-)$ contains precisely all $\sigma$-backward arcs in $T$ with tail in $X$, and $\dirE(I^+,X)$ contains all those with head in $X$. On the other hand, both $\dirE(X,I^-)$ and $\dirE(I^+,X)$ have size bounded by the width of $\sigma_I$, which is at most $\Cctw$. It follows that $X$ is $c_H$-flat with respect to $\sigma$.

As for the algorithmic aspect, an ordering $\sigma_0$ with optimum width can be computed using \cref{lem:sorting}. Then we can divide $\sigma_0$ into $k+1$ disjoint $\sigma_0$-intervals of size at least $(2c+1)(x+1)$ each, and for each of them check in polynomial time whether it induces an $H$-free subtournament using the algorithm from~\cite{FoPi}. If this check fails for each of the intervals, we may terminate the algorithm and conclude that $T$ contains $k+1$ arc-disjoint copies of $H$. Otherwise, we have one interval $I$ such that $T[I]$ is $H$-free, and it is straightforward to turn the remainder of the reasoning into a polynomial-time algorithm computing $\sigma$ and $X$.
\end{proof}

In the remainder of this section let $x\geq 4\|H\|$, $r$, $c$ be fixed positive integers. Moreover, let $T$ be a tournament for which there exists an optimal solution of size at most $k$ and let $\sigma$ be an ordering of $T$. Finally let $X$ be an $H$-free $c$-flat $x(3c+k+1)(2k+3)^{3s(r,c)}$-long $\sigma$-interval. We assume that there is a coloring $\xi$ mapping all $\sigma$-backward arcs incident to $X$ to colors in $[3c]$ such that not two such arcs of the same color participate in the same $\sigma$-cut.

We may apply \cref{cor:Simon-intervals} to $X$ to find a collection of $2k+3$ consecutive $\sigma$-intervals $I_i$ with $|I_i|=x(3c+k+1)$ for all $i\in [2k+3]$ such that all $I_i$ have equal $(r,c)$-signatures, and this common signature, call it $\Sigma$, is equal to the $(r,c)$-signature of their union $I$. Then $I$ is an $H$-free $c$-flat $x(3c+k+1)(2k+3)$-long $\sigma$-interval. That such $I$ can be found in polynomial time (given $r$, $c$, $x$, $T$, $X$, $\sigma$, and $\xi$) follows from \cref{cor:Simon-intervals}.

%In the next lemma we ``dig'' further into $X$ and find an interval $X_\star\subseteq X$ such that we can argue the following: 
Now comes a key step in the proof: we argue that from the equality of types of $I,I_1,\ldots,I_{2k+3}$ it follows that every optimum solution will contain a bounded number of arcs incident with $I$.
%In the next lemma we ``dig'' further into $X$ and find an interval $X_\star\subseteq X$ such that we can argue the following: any optimum solution will contain only a constant number of arcs incident to $X_\star$.

\begin{lemma}\label{lem:finding}
For every optimal solution $F\subseteq A(T)$, we have $|F\cap A(I,V(T))|\leq 2c$.
\end{lemma}

\begin{proof}
Fix an optimal solution $F$. By assumption, $|F|\leq k$.

Define \[F_\star=(F-A(I,V(T)))\cup\sing(I).\]
That is, $F_\star$ is obtained from $F$ by removing all arcs incident with $I$ are replacing them with all $I$-singular arcs. We will prove that $F_\star$ is a solution in $T$. Observe that since $F$ is optimal, this will imply that
\begin{align*}
|F\cap A(I,V(T))|&=|F|-|F-A(I,V(T))|\\&\leq |F_\star|-|F-A(I,V(T))|=|F_\star\cap A(I,V(T))|=|\sing(I)|\leq 2c,
\end{align*} and in consequence we may set $X_\star\coloneqq I$. 

Suppose for the sake of contradiction that $F_\star$ is not a solution in $T$, i.e. there exists an immersion model $\phi$ of $H$ in $T-F_\star$. We may choose $\phi$ so that $\phi|_{I}$ is minimal. As 
\[|I|\geq x(3c+k+1)\geq 4\|H\|(c+(2c+k)+1)\geq 4\|H\|(c+|F_\star|+1),\] by \cref{lem:teleports} we have that $\tau^{(r,c)}(\phi|_{I})$ is short. Note that $\phi$ is incident with $I$ for otherwise it would have not been hit by $F$, contradicting $F$ being a solution. Moreover, $A(\phi)\cap A(I,V(T))\subseteq \gen(I)$  %incidence\mipicom{What is incidence here? Make more precise.} is generic
as all singular arcs incident with $I$ are in $F_\star$ which is disjoint with $\phi$.

Observe that since each of the at most $k$ arcs in $F$ is incident with at most $2$ intervals $I_j$, there exists $i\in [2k+3]$ such that $i\neq 1$, $i\neq 2k+3$, and $I_i$ is not incident with $F$. Pick vertices $v_\alpha\in V(I_1)$, $v_\omega\in V(I_{2k+3})$ which are not incident with $\sing(I)\cup\sing (I_i)\cup F$. This is possible since \[|I_1|,|I_{2k+3}|\geq 3c+k+1>c+c+k\geq|\sing(I)|+|\sing(I_i)|+|F|\geq |\sing(I)\cup\sing (I_i)\cup F|.\]
By equality of signatures of $I$ and $I_i$ and the shortness of $\tau^{(r,c)}(\phi|_{I})$, we conclude that there exists a partial immersion $\phi'$ in $I_i$ of the same type as $\phi|_{I}$, which moreover has a purely generic boundary (i.e. does not contain any $I$-singular arcs).

Now it is enough to observe that we can connect $\phi'$ with $\phi|_{T-I}$ to get an immersion of $H$ in $T$ disjoint with $F$. We shall use $v_\alpha$ and $v_\omega$ as pivots for two-arc generic paths, similarly as in the proof of \cref{lem:teleports}. In other words, to enrich $\phi'$ and $\phi|_{T-I}$ to a full immersion model of $H$, we insert a collection of two-arc generic paths with middle vertices in $v_\alpha$ or $v_\omega$, joining vertices on $\phi'$ with vertices on $\phi|_{T-I}$.%, which  by the definition of $v_\alpha$ and $v_\omega$.

Formally, for every $a\in A(H)$ for which $\phi(a)$ is incident with both $I$ and $T-I$ we do the following. Consider all arcs used for gluing $\phi|_{I}(a)$ with $\phi|_{T-I}(a)$. As argued above, all these arcs are generic. For each such arc $g\in \gen^+(I)$ consider the piece $P$ of the scattered path $\phi(a)$ with $g=\last(P)$ and let $P'$ be the corresponding piece of the corresponding (via equality of the $(r,c)$-types) scattered path $\phi'(a)$. Now replace the arc $\last(P')$  with a two-arc path
\[(\tail(\last(P')),v_{\omega})(v_{\omega},\head(g))\]
joining $P'$ with the piece $Q$ of $\phi|_{T-I}(a)$ which is glued in $\phi$ along $g$ to $P$. Finally remove $g$ from $Q$. We perform a similar replacement for each $g\in \gen^-(I)$ (changing the roles in each pair first/last, head/tail, $v_\alpha$/$v_\omega$). The choice of $v_\alpha$ and $v_\omega$ ensures that the new connections are arc-disjoint and disjoint with $F$.
\end{proof}

We define a digraph $T^\circ$ based on $T$ and $I=I_1\cup\ldots\cup I_{2k+3}$ as follows: start with $T$, and
\begin{itemize}
 \item remove all vertices of $I_2$;
 \item for every arc $a\in \sing^+(I_2)$, replace $a$ with an arc with the same head as $a$ and tail in a fresh vertex  $s^+_{\xi(a)}$;
 \item for every arc $a\in \sing^-(I_2)$, replace $a$ with an arc with the same tail as $a$ and head in a fresh vertex  $s^-_{\xi(a)}$.
\end{itemize}
We call the constructed graph a \emph{$c$-boundaried co-interval}. The intuition of this construction is as follows. We pinch off one of the $2k+3$ intervals and keep the singular arc interface in a fashion similar as in $(r,c)$-boundaried intervals. The only difference is that we do not keep track of the $r$-boundary vertices.

Now we can define the \emph{gluing} of $T^\circ$ with an $(r,c)$-boundaried interval $B=I\cup S^{\pm}$ with signature~$\Sigma$, simply by identifying the singular arcs of the same color (note that different vertices from $S^{\pm}$ can be therefore mapped to the same vertex) and completing the obtained structure to a tournament by making all missing arcs generic. Note that in order for this to be well-defined, we need to require that the sets of colors of the singular arcs in $T^\circ$ and in $B$ are identical --- if it is so, we will say that $T^\circ$ and $B$ are \emph{compatible}. Also, note that we require that the signature of $B$ is $\Sigma$: that is, the possible types of partial immersions present in $B$ are exactly the same as in the substituted interval $I_2$.

Denote by $T^\circ\glue B$ the tournament obtained from gluing $T^\circ$ and $B$. Note that in this tournament we have naturally defined ordering of vertices: in $T^\circ$ it is inherited from the ordering $\sigma$ of $T$ and within $B$ it is inherited from the ordering $\sigma_I$ of the boundaried interval. Finally all the vertices of the substituted interval appear in the ordering between the two interval parts (prefix and suffix) of the co-interval. The following observation is obvious.

\begin{observation}
If two exchangable $(r,c)$-boundaried intervals $B=I\cup S^\pm$ and $B'=I'\cup S^\pm$ are compatible with~$T^\circ$, then the $(r,c)$-signatures of $I$ in  $T^\circ\glue B$ and $I'$ in $T^\circ\glue B'$ are equal.
\end{observation}

We now observe, by inspecting the proof of \cref{lem:finding}, that if in $I$ we replace one of its subintervals with an exchangeable interval, then the conclusion of \cref{lem:finding} --- that the modified $I'$ will still have constant incidence with every optimal solution in $T$ --- will still hold. Let us summarize this section with putting together these observations and recalling all needed assumptions.

\begin{corollary}\label{cor:finding}
Suppose that $x\geq 4\|H\|$, $r$, $c$ are fixed positive integers, $T$ is a tournament of for which there exists an optimal solution of size at most $k$ and $\sigma$ is an ordering of $T$. Suppose further that $X$ is an $H$-free $c$-flat $x(3c+k+1)(2k+3)^{3s(r,c)}$-long $\sigma$-interval with all incident $\sigma$-backward arcs colored according to $\xi$ with colors $[3c]$ so that no two arcs of the same color participate in the same $\sigma$-cut.

Then there exists an $(r,c)$-signature $\Sigma$, a $c$-boundaried co-interval $T^\circ$ of this signature, and an $(r,c)$-boundaried interval $B=I\cup S^\pm$ such that $T=T^\circ\glue B$ and moreover for every $B'=I'\cup S'^\pm$ exchangeable with $B$ and for every optimal solution $F\subseteq A(T')$ where $T'=T^\circ\glue B'$ we have: $|F\cap A(I',T')|\leq 2c$. Moreover, given $x$, $r$, $c$, $T$, $\sigma$, and $X$, such $\Sigma$, $T^\circ$, and $B$ can be computed in polynomial time.
\end{corollary}

\begin{proof}
It is enough to see that the argument from the proof of \cref{lem:finding} can be applied to both $T$ and $T'$ to obtain constant incidence of any optimum solution with $I$ (and $I'$). Constant incidence with an interval means in particular constant incidence with each of its subintervals. As for the algorithmic statement, it follows directly from \cref{cor:Simon-intervals} and \cref{lem:signature-compute}.
\end{proof}

\section{Replacing protrusions}

In the entire section we fix $c\coloneqq \Cctw$, where $\Cctw$ is the constant from \cref{cor:const-ctw}, and $r\coloneqq 6\|H\|c$. Moreover we assume that $x\geq 4\|H\|$; the precise value of $x$ will be determined later.%\mipicom{We probably need more context: we have $T$, some ordering $\sigma$, coloring $\xi$, and so on.}\lbcom{It is introduced below the definition of protrusion.} 

\begin{definition}
\emph{Protrusion} is any $(r,c)$-boundaried interval $X=I\cup S^\pm$ %\mipicom{Could we please just remove $\Sigma$ from the definition of a boundaried interval, and define the signature of a boundaried interval? For instance, here we use the signature of a boundaried interval, which is undefined.} 
which is $H$-free and $x(3c+k+1)$-long. For brevity we will refer to $X$ as to $I$. The set $\Sigma^{(r,c)}(X)$ is the \emph{signature} of the protrusion.
\end{definition}

Let $T$ be a tournament equipped with a vertex ordering $\sigma$ and let $I$ be an $H$-free interval and such that $T=T^\circ\glue X$, where $X$ is a protrusion of signature $\Sigma$. Let $I\subseteq V(T)$ be the interval defined by the protrusion. Fix the coloring $\xi$ of $\sigma$-backward arcs incident with $I$. 

Recall that from \cref{cor:finding} follows that for every $(r,c)$-boundaried co-interval $T^\circ$ compatible with $X$ if $T^\circ\glue X$ admits an optimal solution $F$ of size not greater than $k$, then there are at most $2c$ arcs in $F$ incident with $X$.

For every protrusion $X=I\cup S^\pm$ define a function $f_X\colon 2^{\S^{(r,c)}}\to \{0,1,2,\ldots,2c\}\cup\{\infty\}$ as follows: $f_X(S)$ is the minimum number of arcs in $A(I,I)$ needed to hit all partial immersions in $I$ whose signatures belong to $S$, or $\infty$ if this number is greater than $2c$. 

We introduce an equivalence relation $\sim$ on the set of all protrusions. Let $X=I\cup S^\pm$, $X'=I'\cup S'^\pm$. We say that $X\sim X'$ if:
\begin{itemize}
 \item $S^\pm=S'^\pm$; %the sets of $\xi$-colors\mipicom{Formally, a protrusion has no coloring $\xi$} of $\sigma$-backward arcs appearing in the boundaries of $X$ and of $X'$ are equal;
 \item $\Sigma^{(r,c)}(X)=\Sigma^{(r,c)}(X')$; and
 \item 
$f_X(S)= f_{X'}(S)$ for every $S\subseteq \S^{(r,c)}$.
\end{itemize}
Note that the number of equivalence classes of $\sim$ is finite and bounded by a constant depending only on~$H$. This means that if in each class we pick a representative with the minimal number of vertices (call each such element a \emph{small} protrusion), then all small protrusions will have size bounded from above uniformly by a constant $s_{H}$ depending on the digraph $H$ only.

Same arguments as in the proof of \cref{lem:signature-compute} give the following.

\begin{lemma}\label{lem:protr-equiv-algo}
 Given protrusions $X$ and $X'$ it can be decided in polynomial time whether $X\sim X'$.
\end{lemma}

We now argue that equivalent protrusions are replaceable.

\begin{lemma}\label{lem:replacing}
Suppose that $T=T^\circ\glue X$ is a tournament satisfying the conclusion of \cref{cor:finding}, where $X$ is a boundaried $(r,c)$-interval of signature $\Sigma$ and $T^\circ$ is a $c$-boundaried co-interval which is compatible with $X$.

Then for every $X'$ such that $X\sim X'$, the optimal solution in $T^\circ\glue X$ is of size not greater than $k$ if and only if the optimal solution in $T^\circ\glue X'$ is of size not greater than $k$. 
\end{lemma}

\begin{proof}
We will prove the forward implication only; the proof of the other one is completely analogous due to the symmetry of the roles of $X$ and $X'$.

Let $F$ be an optimal solution in $T$ and suppose that $|F|\leq k$. Let 
$$F_0=F - A(X,V(T)),\qquad F_1=F\cap A(V(T)\setminus X,X)\quad\textrm{and}\quad F_X=F\cap A(X,X).$$ By the assumption that $T=T^\circ \glue X$ satisfies the conclusion of \cref{cor:finding}, we have $|F_1\cup F_X|\leq 2c$.

\begin{claim}\label{cl:only-border}
Every $X$-generic arc of $F_1$ is incident with the $r$-boundary of $X$.
\end{claim}

\begin{proof}[of claim]
Suppose for the sake of contradiction that in $F_1$ there exists an $X$-generic arc $g$ not incident with the $r$-boundary of $X$. This means that this arc is incident with an \emph{internal} vertex of $X$, i.e. the vertex being $\sigma$-smaller than all vertices in $X_r^+$ and $\sigma$-larger than all vertices in $X_r^-$.

By the optimality of $F$ we conclude that in $T-(F-g)$ there exist immersion models of $H$, and moreover each such model uses the arc $g$. Let us choose such model $\phi$ with the property that the partial immersion $\phi|_{X}$ cannot be non-trivially shortened. We will show a way to ``reroute'' the generic connection containing $g$ to avoid $F$, all $X$-singular arcs and all other generic arcs used in $\phi$. This way we expose an immersion model of $H$ in $T-F$, which is a  contradiction.

This is the part of the proof where the $r$-boundary of $X$ (and keeping track of the $r$-boundaries in the definition of types) comes into play. In the beginning of this section we picked $r$ large enough to ensure the existence of a ``pivot'' vertex, which can be used for rerouting.

Suppose that $g\in\gen^+(X)$; the case $g\in\gen^-(X)$ is analogous. Denote $v=\tail(g)$. Note that at most $4c$ vertices of $X_{r}^+$ are incident with $F_1\cup F_X$, as $|F_1\cup F_X|\leq 2c$. Moreover $|\dirE(X_r^+,\{v\})|\leq c$, because $X$ is $c$-flat. The number of vertices of $X_{r}^+$ incident with $\sing^-(X)$ also does not exceed $c$.

For every $a\in A(H)$ the number of arcs on $\phi|_{X}(a)$ with tail in $v$ does not exceed $1$ as otherwise $\phi|_{X}(a)$ could be non-trivially shortened contradicting minimality of $\phi|_{X}$. So the number of arcs in $\gen^+(X)\cap(F_1\cup F_X\cup A(\phi|_{X}))$ incident with $X_r^+$ is not greater than $\|H\|$. Together this means that if
\[r> 6c+\|H\|,\]
which is true, then there exists a vertex $v_{\omega}\in X_r^+$ which is not incident with any of the special arcs mentioned above. The arc $(v,v_{\omega})$ exists and is not contained in $F\cup A(\phi)$. Neither is the arc joining $v_{\omega}$ with the respective vertex in $T-X$. The arc $g$ may therefore be replaced by a path composed of the two arcs described above. And so we have obtained an immersion model of $H$ in $T-F$, a contradiction.

\end{proof}

Let $S$ be the set of all short types $\tau$ such that every $X$-admissible partial immersion of type $\tau$ intersects with $F_X$. Because $f_X(S)=f_{X'}(S)$, we know that there exists a set $F_{X'}$ of arcs of $A(X')$ in $T'\coloneqq T^\circ\glue X'$ such that $|F_{X'}|\leq |F_X|$ and each partial immersion whose type is in $S$ intersects with $F_{X'}$. Let $F'_1\subseteq A(T'-X',X')$ be the set directly corresponding to $F_1$ by a natural bijection that maps each singular arc of to the singular arc of the same color, and each generic arc to a generic arc whose the endpoint in $T^\circ$ is the same and the other endpoint has the same index in the respective $r$-boundary (by \cref{cl:only-border}, this endpoint must belong to the respective boundary). We will prove that the set $F'\coloneqq F_0\cup F'_1\cup F_{X'}$ is a solution in $T'$, thereby completing the proof due to $|F'|\leq |F|$.

Suppose for the sake of contradiction that there exists an immersion model $\phi$ of $H$ in $T'-F'$. Choose one such that $\phi|_{X'}$ is minimal. Therefore, by \cref{lem:teleports}, $\phi|_{X'}$ is short. Moreover $\tau^{(r,c)}(\phi|_{X'})\notin S$, so in $T$ there exists a partial immersion $\phi_X$ in $X$ which is disjoint with $F$ and such that $\tau^{(r,c)}(\phi_X)=\tau^{(r,c)}(\phi|_{X'})$. We will show a way to reconstruct an immersion model of $H$ in $T$ from $\phi_X$ and $\phi|_{T^\circ}$ which will be disjoint with $F$ and hence will give the desired contradiction. We may follow closely the strategy used in the proof of \cref{lem:types-behave-well}, taking an additional care about ensuring disjointness with $F$.

We know that $\phi_X$ is disjoint with $F$, so for sure the direct gluing along singular arcs is possible, as it was so with $\phi|_{X'}$ and $\phi_{T^\circ}$. Moreover, gluing along generic arcs with one endpoint $u$ on the $r$-boundary of $X$ and the other one $v$ in $T^\circ$ is possible since the arc $uv$ is generic, not incident with $F$ and will not have to be introduced more than once (this is ``guarded'' by the equivalence relations $R^{\pm}$).

\end{proof}

We are now set up with all tools needed to prove our main theorem.

\begin{proof}[of \cref{thm:main}]
We prove that provided $|T|> C\cdot k^C$, where the constant $C$ will be defined later, one can compute an instance $(T',k)$ equivalent to $(T,k)$ and satisfying $|T'|<|T|$. The conclusion will follow from applying such reduction (at most) $|T|$ times.

First of all note that from \cref{thm:ctw-bound} for a digraph being a disjoint union of $k$ exemplars of $H$, we conclude that if $T$ does not contain $k$ arc-disjoint immersion copies of $H$, then %$\ctw(T)\in \O_{H}(k^2)$, say 
$\ctw(T)\leq c_0k^2$ for some constant $c_0$ (depending on $H$). In particular, it follows that if $\ctw(T)> k^2 c_0$ (which can be established in polynomial time using \cref{lem:sorting}), then $T$ is a {\sc no}-instance. So from now on we may assume that $\ctw(T)\leq c_0k^2$.

Let $c=\Cctw$, $r=6\|H\|c$, $x' = \max\left\{4\|H\|,s_H+1\right\}$ and 
$x=x'(3c+k+1)(2k+3)^{3s(r,c)},$ 
where $s(r,c)$ is defined in \cref{cor:const-sign}. Let $C$ be a constant satisfying $C k^C \geq (2k^2c_0+1)(x+1)(k+1)$, e.g.
$C=\max\{3s(r,c)+4,5^{3s(r,c)}\cdot 3c_0\cdot 4x'\cdot (6c+2)\}$.

Suppose that $T$ is a tournament satisfying $\ctw(T)\leq k^2c_0$ and $|T|>C k^C$. Applying \cref{lem:protrusion-const-boundary}, we either conclude that $T$ admits more than $k$ arc-disjoint copies of $H$ (so $(T,k)$ is a {\sc no}-instance), or find an ordering $\sigma$ of $V(T)$  and an $H$-free $c$-flat $x$-long $\sigma$-interval $J$. Both conclusions can be effectively gained in polynomial time.

In the latter case, we may use \cref{cor:Simon-intervals} in a manner described in \cref{sec:finding} to find in $J$ an $H$-free $c$-flat $x'(3c+k+1)(2k+3)$-long $\sigma$-interval $I$ of $(r,c)$-signature $\Sigma$, which can be decomposed to $2k+3$ consecutive $x'(3c+k+1)$-long $\sigma$-intervals $I_i$, each of $(r,c)$-signature $\Sigma$. Both $I$ and $\Sigma$ are found in polynomial time.
%Let $r=6\|H\|c$, $c=\Cctw$, $x'=4\|H\|$ and $x=x'\cdot(3c+k+1)(2k+3)^{3s(r,c)}\in k^{\O_H(1)}$, where $s(r,c)$ is defined in \cref{cor:const-sign}. If $|T|<(2c+1)(x+1)(k+1)$, then we put $T'=T$, which obviously leads to the desired  equivalent instance of size bounded polynomially in $k$. Otherwise applying \cref{lem:protrusion-const-boundary}, we find an ordering $\sigma$ of $V(T)$ of width at most $\ctw(T)$ and an $H$-free $x$-long $c$-flat $\sigma$-interval $X$.

Let $T=T^\circ\glue X$ be the decomposition where $T^\circ$ is a $c$-boundaried co-interval and  $X=(I_2\cup S^\pm,\Sigma)$ is a protrusion corresponding to $I_2$. Clearly, $X$ is compatible with $T^\circ$. Using \cref{lem:protr-equiv-algo} we may check in polynomial time all small protrusions to find one $X'$ such that $X'\sim X$. Let us define $T'=T^\circ\glue X'$.

By \cref{lem:replacing} we conclude that $(T',k)$ is an instance of \himmt{} equivalent to $(T,k)$. Moreover as $|X'|\leq s_H<|X|$, we have that $|T'|<|T|$. 
\end{proof}

\pagebreak[3]
\bibliographystyle{abbrv}
\bibliography{main}

\newpage
\begin{appendices}
  \crefalias{section}{appsec}
\section{Derivation of \texorpdfstring{\Cref{lem:Simon}}{Lemma 5}}\label{sec:appa}

Given an alphabet $A$ and a word $w\in A^\star$, we say that $w$ can be \emph{decomposed} into words $w_1$, $w_2$, $\ldots$, $w_\ell$ if $w=w_1w_2\ldots w_\ell$. An unranked rooted tree $T$ labeled with elements of $A^\star$ is a \emph{decomposition tree} of $w$ if:
\begin{itemize}[nosep]
\item the root is labeled with $w$;
\item the label of every non-leaf node can be decomposed into labels of all sons of this node (from left to right);
\item the label of every leaf node is an element of $A\cup\{\varepsilon\}$ (here $\varepsilon$ is the empty word).
\end{itemize}
The \emph{degree} of a node of a decomposition tree is the number of its sons. Given a monoid $S$ and a morphism $\alpha\colon A^\star\to S$, we say that decomposition $w=w_1w_2\ldots w_\ell$ satisfies the \emph{$\alpha$-idempotent rule} if there exists an idempotent $e\in S$ such that
\[\alpha(w_1)=\alpha(w_2)=\ldots=\alpha(w_\ell)=e.\]  We define an \emph{$\alpha$-factorization forest} for a word $w\in A^\star$ as a decomposition tree of $w$, whose each non-leaf node with degree greater than $2$ corresponds to an $\alpha$-idempotent rule.

\begin{theorem}[Factorization Forest theorem of Simon]\label{thm:Simon}
Let $S$ be a finite monoid, $A$ be a finite alphabet, and $\alpha\colon A^\star\to S$ be a morphism. Then every  $w\in A^\star$ has an $\alpha$-factorization forest of height at most $3|S|$.
\end{theorem}

The intuition behind \Cref{lem:Simon} is as follows. If word $w$ is appropriately long then existence of a constant-height factorization forest forces existence of a node of large degree (hence corresponding to an idempotent rule) in this forest. 

\begin{proof}[of \cref{lem:Simon}]
Fix $w\in A^\star$ of length at least $\ell^{3|S|}$. By \cref{thm:Simon} there exists an $\alpha$-factorization forest $T$ for $w$ of height at most $3|S|$. Note that if $T$ admits a node of degree at least $\ell$, then the idempotent rule corresponding to this node provides the desired subword $w'$ and its decomposition into $\ell$ non-empty subwords (in the case of degree strictly larger than $\ell$, we may concatenate all subwords whose indices are not smaller than $\ell$ into a single subword).

We will prove that such node always exists. Suppose for the sake of contradiction that all non-leaf nodes of $T$ have degree strictly smaller than $\ell$. As $T$ has height at most $3|S|$, the number of leaves of $T$ is therefore strictly smaller than $\ell^{3|S|}$. On the other hand, each letter of $w$ has a different corresponding leaf node of $T$, so there are at least $|w|\geq \ell^{3|S|}$ leaf nodes. The contradiction finishes the proof.
\end{proof}

\end{appendices}

\end{document}